\documentclass [10pt,twoside]{article} 
\usepackage{times} 
\pdfoutput=1

\usepackage{amsmath}
\usepackage{amsthm}
\usepackage{amssymb}
\usepackage{amsfonts}
\usepackage[all]{xy}
\usepackage{epsfig}   
\usepackage{color}
\usepackage{ifthen}
\usepackage{verbatim}
\usepackage{pstricks}
\usepackage[numbers, sort&compress]{natbib}
\usepackage{appendix}
\usepackage{hyperref}

\def\bR{\begin{color}{red}}
\def\bB{\begin{color}{blue}}
\def\bM{\begin{color}{magenta}}
\def\bC{\begin{color}{cyan}}
\def\bW{\begin{color}{white}}
\def\bBl{\begin{color}{black}} 
\def\bG{\begin{color}{green}}
\def\bY{\begin{color}{yellow}}
\def\e{\end{color}}

\theoremstyle{definition}\newtheorem{remark}{Remark}
\theoremstyle{definition}\newtheorem{example}{Example}

\newcommand{\bit}{\begin{itemize}}
\newcommand{\eit}{\end{itemize}\par\noindent}
\newcommand{\ben}{\begin{enumerate}}
\newcommand{\een}{\end{enumerate}\par\noindent}
\newcommand{\beq}{\begin{equation}}
\newcommand{\eeq}{\end{equation}\par\noindent}
\newcommand{\beqa}{\begin{eqnarray*}}
\newcommand{\eeqa}{\end{eqnarray*}\par\noindent}
\newcommand{\beqn}{\begin{eqnarray}}
\newcommand{\eeqn}{\end{eqnarray}\par\noindent}

\theoremstyle{plain}
\theoremstyle{plain}
\newtheorem{thm}{Theorem}
  \newtheorem{defn}[thm]{Definition}
	\newtheorem{lem}[thm]{Lemma}
	\newtheorem{cor}[thm]{Corollary}
	\newtheorem{prop}[thm]{Proposition}
	\newtheorem{post}[thm]{Postulate}

\newcommand{\xyR}[1]{%
\xydef@\xymatrixrowsep@{#1}}
\newcommand{\xyC}[1]{%
\xydef@\xymatrixcolsep@{#1}}

\newcommand{\I}{{I}}
\newcommand{\superscript}[1]{\ensuremath{^{\textrm{#1 }}}}
\newcommand{\st}[0]{\superscript{st}}
\newcommand{\nd}[0]{\superscript{nd}}
\newcommand{\rd}[0]{\superscript{rd}}
\newcommand{\rth}[0]{\superscript{th}}

\title{\mbox{Causal categories: relativistically interacting processes}}
\author{Bob Coecke and Raymond Lal} 
\date{University of Oxford, Computer Science, Quantum Group,\\ 
Wolfson Building, Parks Road, Oxford OX1 3QD, UK.\\
{\normalsize\tt coecke/rayl@cs.ox.ac.uk}} 
\begin{document}
\maketitle
\thispagestyle{empty}

\begin{abstract}
A symmetric monoidal category naturally arises as the mathematical structure that organizes physical systems, processes, and composition thereof, both sequentially and in parallel.  This structure admits a purely graphical calculus. This paper is concerned with the encoding of a fixed causal structure within a symmetric monoidal category: causal dependencies will correspond to topological connectedness in the graphical language.
We show that correlations, either classical or quantum, force terminality of the tensor unit. We also show that well-definedness of the concept of a global state
forces the  monoidal product  to be only partially defined, which in turn results in a relativistic covariance theorem.  Except for these assumptions, at no stage do we assume anything more than purely compositional symmetric-monoidal categorical structure.  We cast these two structural results in terms of a mathematical entity, which we call a \em causal category\em.  We provide methods of constructing causal categories, and we study the consequences of these methods for the general framework of categorical quantum mechanics.
\end{abstract}

%


\tableofcontents
\
\section{Introduction}


This paper is concerned with the causal structure of fundamental theories  of physics.
We cast  causal aspects of both relativity and    quantum theory in a unified setting, that is, a  single mathematical entity that will be derived from certain phenomenological considerations  from each theory. Our starting point is \em categorical quantum mechanics \em (CQM) \cite{AC1}, a general framework for physical theories in which \em type\em, \em process\em, and \em composition \em thereof,  are the primary concepts \cite{ContPhys, Templeton}.  The two modes of composing processes, \em in  parallel \em and \em sequentially\em, already provide an imprint of causal structure, admitting the interpretation of temporal and spatial composition respectively.  We seek to  make 
this correspondence  more precise, 
such that one can encode a fixed causal structure within the category.  From a dual perspective, we `thicken' a causal structure \cite{KP, Sorkin, KeyePrakash} to a proper category,  so that we obtain a category that  
encodes more than just causal relationships,  but which 
also encodes the processes that may take place along these causal connections.
 
We draw on earlier work by Markopoulou \cite{Markopoulou}, Blute-Ivanov-Panangaden \cite{BIP}, Hardy \cite{HardyTempleton, HardyPicturalism} and Chiribella-D'Ariano-Perinotti \cite{CDP1, CDP2}.\footnote{ In turn, the work by Hardy in  \cite{HardyTempleton, HardyPicturalism} and Chiribella-D'Ariano-Perinotti in \cite{CDP1, CDP2} is strongly influenced by CQM; in particular, by taking a diagrammatic language for processes as their starting point.}  In these papers, with increasing levels of abstraction, one considers a diagram representing causal connections, and decorates it with specific quantum events or processes.  We trim the assumptions in this work down to their `bare categorical bones', while retaining the key results:
\bit
\item a covariance theorem for global states as in \cite{BIP, HardyTempleton};
\item uniqueness of effects of a certain type as in \cite{CDP1}.
\eit
In particular, unlike the previous work mentioned above, our derivation does not make any reference to measurement or probabilities, just to a very general concept of process, and hence is more primitive.  We then observe that at the most basic level of the causal structure, the quantum-mechanical structure itself does not  play a crucial role; for example, our structure also captures classical  probabilistic processes that take place in relativistic space-time. 

The resulting structure is one which organizes processes which  can  potentially take place within a causal setting, together with their compositional interaction.  This allows, for example, the collection of  possible physical processes 
to vary according to 
the causal structure, for example, taking into account the different capabilities of distinct agents, or differences of a purely physical origin.  

Conceptually speaking, the stance of elevating processes to a privileged role in theories of physics was already present in the work of Whitehead in the 1950s \cite{Whitehead} and the work of Bohr in the early 1960s \cite{Bohr}.  It  became more prominent in the work of Bohm in the 1980s and  also  later in Hiley's \cite{Bohm, BohmHiley1, BohmHiley2}, who is still pursuing this line of research  \cite{Hiley}.  It is an honor to dedicate this paper to Basil Hiley, on the occasion of his 75th birthday.


\paragraph{Plan of the paper.}
In Section \ref{sec:CQM} we give an overview of symmetric monoidal categories and CQM, and we describe the problem that partly motivates this paper, namely how compact structure, interpreted as post-selected quantum teleportation  \cite{AC1, ContPhys}, leads to signaling. In Sections \ref{sec:Derivation} and \ref{sec:partial} we show how consideration of this problem and other phenomenological issues leads to the causal category structure. In Section \ref{sec:Definition} we formally define causal categories and describe some of their basic properties, in particular the way in which causal categories are incompatible with some structures of CQM. In Section \ref{sec:constructCC} we define methods of constructing causal categories, and how key features of CQM are retained.



%
%

\section{Processes as pictures}\label{sec:CQM}

In this Section we describe the existing framework for doing \em categorical quantum mechanics \em (CQM) using symmetric monoidal categories, and we state the problem addressed in this paper.  


\subsection{Symmetric monoidal categories}\label{sec:monoidal} 

Symmetric monoidal categories are mathematical entities with a direct physical interpretation;  introductions to the subject are \cite{AbrTze, BaezStay, CatsII}. Their role in CQM is that they provide two modes for composing systems and processes, sequential and concurrent. More precisely put:

\paragraph{1. Entities.} We shall consider a  collection of named \em objects \em or \em systems \em $A, B, C, ...$, and \em morphisms \em or \em processes \em $f:A\to B$ which may take a system of one kind $A$ into a system of another kind $B$. 
We call $A\to B$ the \em type\em\footnote{The term `type' reflects the application of category theory to theoretical computer science \cite{AspertiLongo}.} of the process $f$, for which $A$ is the \em domain \em and $B$ is the \em codomain\em. The set of all processes taking $A$ into $B$ is denoted ${\bf C}(A,B)$.  
States correspond to  processes from a special object ${I}$ into $A$, where one may interpret  ${I}$ as the `unspecified environment'.  From an operational perspective one can think of such a process as a preparation procedure, while from an ontic perspective one can think of it as the unknown process that caused $A$ to be in this state. Moreover, effects correspond to  processes from an object $B$ to  ${I}$, and  \em scalars \em or \em weights \em correspond to processes from  ${I}$ to  ${I}$.

\paragraph{2. Graphical language.} We now introduce a graphical language to represent our entities \cite{ContPhys}: systems are  represented by wires, and  processes $f:A\to B$ are represented by boxes with an input wire representing system $A$ and an output wire representing system $B$:
\[
\raisebox{-10mm}{\epsfig{figure=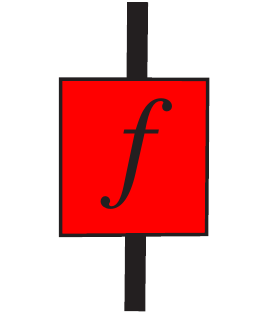,width=28pt}}
\]
In fact, as we shall discuss below, there may be more than one input and output wire---which could result from composing systems---or indeed no input and/or no output wires, representing `no system'---denoted above by $I$.

\paragraph{3. Composition.} The mathematical content of the formalism is given by the composition of processes. There are two  \em connectives \em which allow the composition of processes both `in parallel' and `sequentially':   
\bit
\item The \em sequential\em, or \em dependent\em, or \em causal\em, or \em connected \em composition of processes $f:A\to B$ and $g:B\to C$ is  $g\circ f:A\to C$,  and is depicted as:
\beq\label{def:seqcomp}
\raisebox{-10mm}{\epsfig{figure=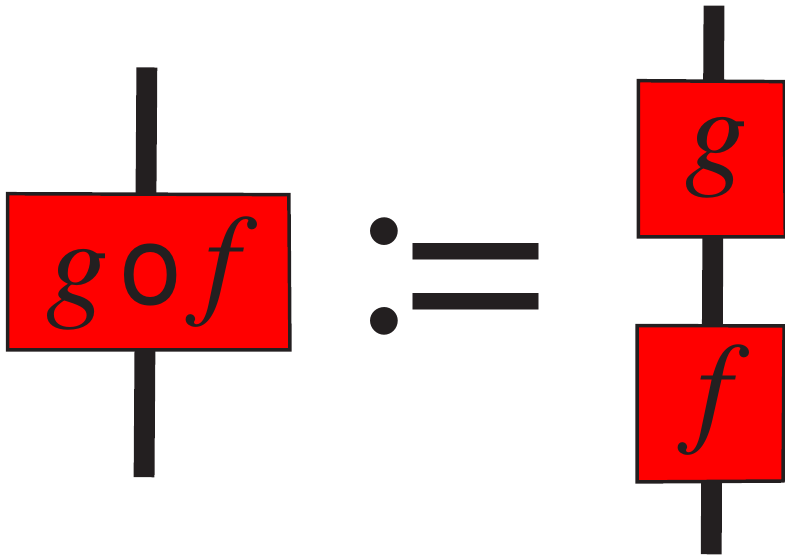,width=80pt}}\quad\ \ 
\eeq
\item The \em parallel\em, or \em independent\em, or \em acausal\em, or \em disconnected \em composition of processes $f:A\to B$ and $g:C\to D$  is $f\otimes g:A\otimes C\to B\otimes D$, 
and is depicted as:

\beq\label{def:parcomp}
\raisebox{-4mm}{\epsfig{figure=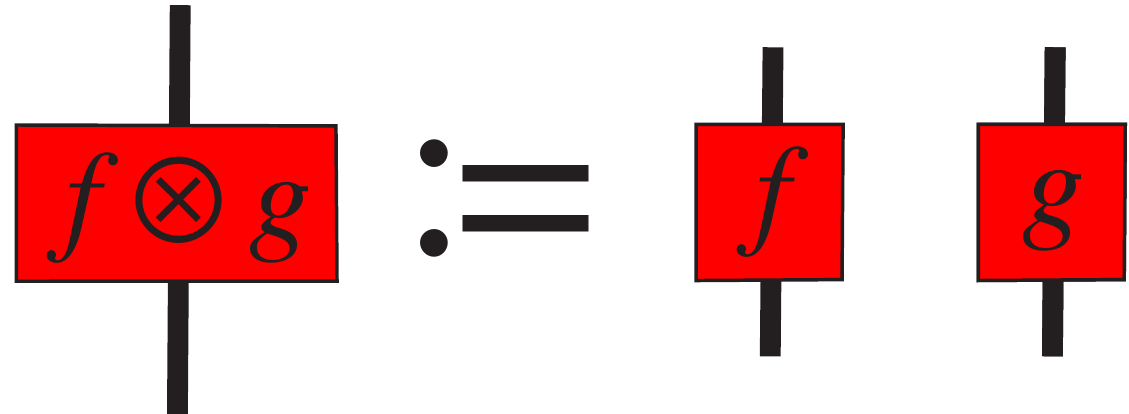,width=115pt}}\quad\ \  
\eeq
\eit

Importantly, compoundness of physical systems is now directly apparent in the graphical notation, in that there can be several wires side-by-side:

\begin{center}
\ one system\qquad\ two systems  \ \ \qquad  $n$ systems\ \ \qquad \ \ 
\raisebox{-2mm}{\begin{minipage}[b]{18mm} 
\begin{center}
operation on  $n$ systems
\end{center} 
\end{minipage}}
 \\ \ \\
\epsfig{figure=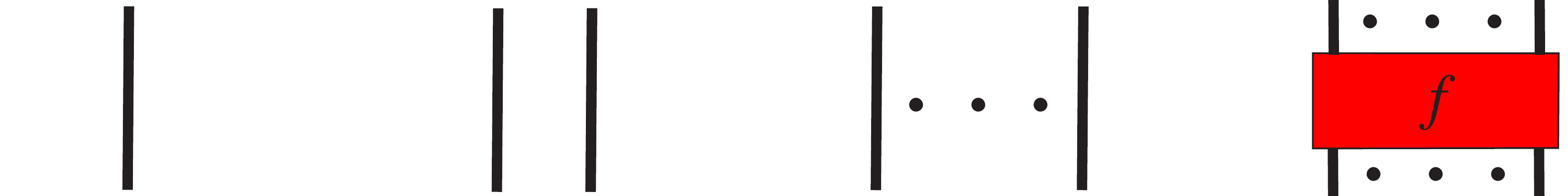,width=260pt} 
\end{center}

Note that there is an imprint of causal structure in this formalism 
(as indicated in the terminology), 
since one can think of the `acausal' composition $\otimes$ as `spatially' separating, while one thinks of the `causal' composition $\circ$ as `temporally' connecting 
\[
\epsfig{figure=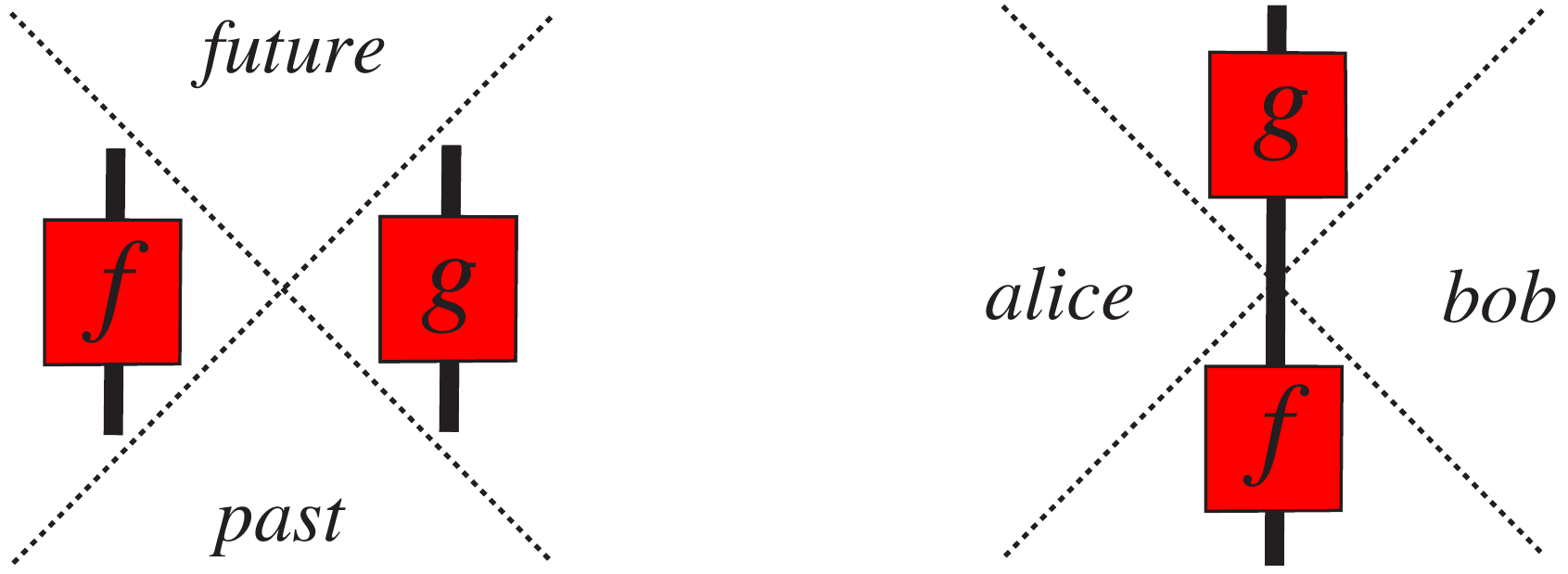,width=172pt}\vspace{-1.5mm}
\]

We now give the symbolic definition of a symmetric monoidal category. Here we  restrict to the case of  \em strict \em symmetric monoidal categories, since that is what the diagrammatic calculus embodies. More importantly, as argued in \cite{CatsII},  physical processes always form strict symmetric monoidal categories.\footnote{On the other hand, their mathematical representations typically form non-strict symmetric monoidal categories; see again \cite{CatsII} for a discussion of this point.}

\begin{defn}\em
A \em strict monoidal category \em is a category $\mathbf{C}$ equipped with a bifunctor $\otimes:\mathbf{C}\times\mathbf{C}\rightarrow\mathbf{C}$ and a \em unit object ${I}$ \em such that
$(\vert \mathbf{C} \vert , \otimes, I)$ is a monoid, and for all $f,g,h$ in $\mathbf{C}$ we have associativity of acausal composition:
\beq
f\otimes (g\otimes h)=(f\otimes g)\otimes h\,.
\eeq

\end{defn}

The content of $\otimes$ being a bifunctor is that there is an interchange law between $\otimes$ and $\circ$: for all morphisms $f,g,h,k$ (of appropriate type) in a strict monoidal category, we have:
\beq\label{eq:bifunct}
(g\circ f)\otimes (k\circ h) = (g\otimes k)\circ(f\otimes h)\,.
\eeq

Now, recall that an isomorphism between two objects $A$ and $B$ in a category is  a morphism $f:A\rightarrow B$ for which there exists an \em inverse \em  $f^{-1}:B\rightarrow A$, that is, a morphism which is  such that $f^{-1}\circ f=1_{A}$ and $f\circ f^{-1}=1_{B}$,  where $1_A:A\to A$ denotes the \em identity morphism \em on $A$. 

\begin{defn}\em 
A \em symmetric strict  monoidal category \em (SMC) is a monoidal category $\mathbf{C}$ equipped with a family of  symmetry isomorphisms  $\sigma_{A,B}:A\otimes B\rightarrow B\otimes A$ which is such that for all $A,B\in\vert \mathbf{C} \vert$ we have $\sigma^{-1}_{A,B}=\sigma_{B,A}$, and for all $A,B,C,D\in\vert \mathbf{C} \vert$ and $f:A\rightarrow C,g:B\rightarrow D$ in $\mathbf{C}$, we have:
\beq\label{eq:symnat}
\sigma_{C,D}\circ(f\otimes g)=(g\otimes f)\circ \sigma_{A,B}.
\eeq
\end{defn}

\begin{example}[Programming languages and proof theory]
Symmetric monoidal categories play a significant role in the theory of programming languages and proof theory, a modern branch of logic.  In programming, objects represent data types and a morphism $f:A\to B$ stands for running a program that requires input data of type $A$ and produces output data of type $B$. Sequential composition would then mean first running one program and then using its output as the input for the second program.  Parallel composition means running two programs in parallel.  In proof theory, objects represent propositions and a morphism $f:A\to B$ stands for a derivation of proposition $B$ given proposition $A$.
\end{example}

\begin{example}[Dirac notation] 
Morphisms for which the domain and/or the co-domain is the tensor unit $\I$ have a special form in the graphical notation. A generic  \em element \em or \em state \em  $\psi:\I\to A$ (cf.~a Dirac `ket'  $|\psi\rangle$ in quantum mechanical Dirac notation) is depicted as $\raisebox{-1mm}{\epsfig{figure=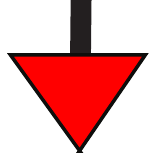,width=12pt}}\,$, and a generic  \em co-element \em or \em effect \em $\psi:  \I \to A$ (cf.~a Dirac `bra'  $\langle\psi|$) as $\raisebox{-1mm}{\epsfig{figure=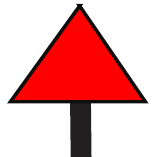,width=12pt}}\,$ \cite{ContPhys}.  This shows how the graphical language is a `two-dimensional' extension of Dirac notation; consider the graphical notation for states and effects compare to Dirac bras and kets:

\begin{center}
\begin{tabular}{|c|c|} 
 $|\psi\rangle$  & $\raisebox{-1mm}{\epsfig{figure=CCstate.pdf,width=12pt}}\,$  \\ 
\hline
$\langle\psi|$    &  $\raisebox{-1mm}{\epsfig{figure=CCeffect.pdf,width=12pt}}\,$ \\
\end{tabular} 
\end{center}
We notice that a clockwise rotation of the Dirac ket by $90\,^{\circ}$ yields the same triangle shape as the graphical notation on the right hand-side; and similarly for bras.
 Moreover, given a state $\psi:{I}\rightarrow A$ and an effect $\phi:A\rightarrow {I}$, we obtain a morphism with the `trivial' type $\psi\circ\phi: {I}\to {I}$. As mentioned above, this is a scalar, and is denoted graphically as having no input nor output wires, which is again a rotated denotation of $\langle \phi\mid\psi\rangle$. 
\end{example}

The symmetry natural isomorphism is denoted graphically by a crossing, so Eq.~(\ref{eq:symnat}) is depicted as:
\beq\label{eq:symnat2}
\raisebox{-6mm}{\epsfig{figure=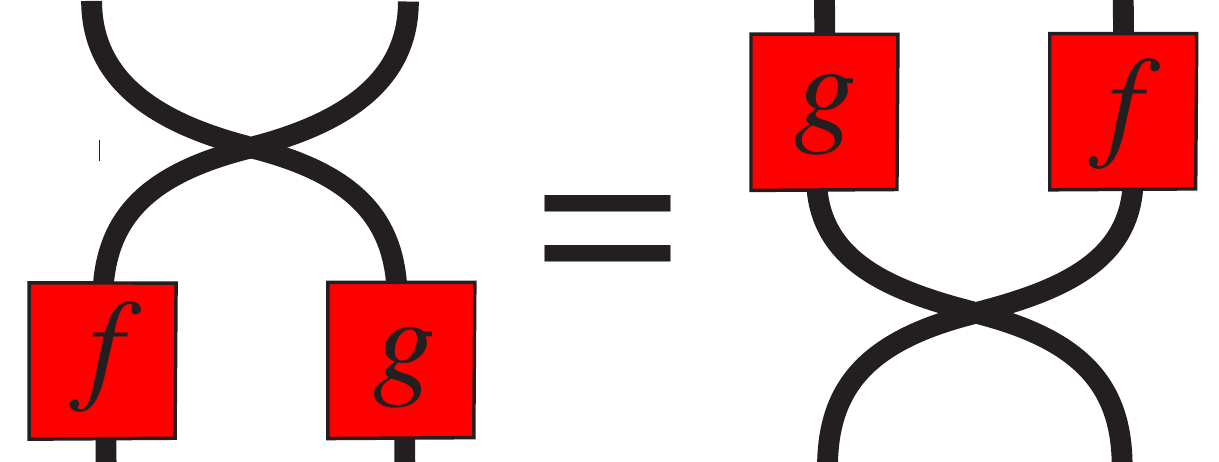,width=120pt}}
\eeq
This gives an indication of how the graphical \em calculus \em will subsume symbolic equations:  
Eq.~(\ref{eq:symnat2}) captures  Eq.~(\ref{eq:symnat}) by  the intuitive notion of `sliding boxes along a wire'.   It will therefore be useful to formally distinguish graphical and symbolic representations.

\begin{defn}[symbolic language]\label{def:formula}\em
By an \em object formula \em in the (symbolic) language of an SMC we mean any expression involving objects and the tensor thereof. By a \em (well-formed) morphism formula \em in the (symbolic) language of an SMC we mean any  expression involving morphisms, sequential composition, and parallel composition thereof, for which sequential composition only occurs for morphisms with matching types.
\end{defn}

Consider an object formula $A_1\otimes A_2$ in a category $\mathbf{C}$, with $A=A_1\otimes A_2$. We shall be careful to distinguish between $A$ and $A_1\otimes A_2$, since there may be other objects, say $B_1$ and $B_2$,  such that $A=B_1\otimes B_2$, and hence, the object formula $A_1\otimes A_2$
contains more information than its corresponding object in the
category $A$ does, namely,  it shows how it was formed. The same applies to morphism formulae, e.g.~$(1_A \otimes k) \circ h\circ (f\otimes g)$, which again contains more information than the corresponding morphism in the category.  
We shall notationally distinguish the object language and objects as follows:
\bit
\item Object formulae will be denoted by calligraphic capital letters ${\cal A}, {\cal B}, {\cal C}, \dots$
\item Objects will be denoted by Roman-font capital letters $A,B,C,\dots$
\eit
and morphism language and morphisms as follows:
\bit
\item Morphism formulae will be denoted by calligraphic capital letters ${\cal F}, {\cal G}, {\cal H}, ...$
\item  Morphisms will be denoted by Roman font $f, g, h, ...$ 
\eit
Each morphism formula ${\cal F}$ has an object formula as its input and output, which we specify by writing  ${\cal F}:{\cal A}\to {\cal B}$. We can associate to $\mathcal{F}$ a \em corresponding morphism \em $f:A\to B$, simply by performing the compositions expressed within the object formulae and morphism formula. We define \em corresponding objects \em for object formulae similarly. We write $A:={\cal A}$, $B:={\cal B}$ and  $f:={\cal F}$, a notational convention which we already use  in Eqs.~(\ref{def:seqcomp}) and (\ref{def:parcomp}) above, where the right-hand side represents the diagram expressing the composition, while the left-hand side represents the morphism that one obtains when performing the composition. We use the notation ${\cal F}:A\to B$ to mean that in ${\cal F}: {\cal A}\to {\cal B}$ we have $A:={\cal A}$ and, $B:={\cal B}$. An equation ${\cal F}={\cal G}$ means that the corresponding morphisms are equal, i.e.~$f=g$. The physical intuition behind this is that several `physical scenarios' or `experimental protocols', while  distinct in their implementation details, may have exactly the same overall effect.

The graphical elements we have introduced correspond formally to the entities defined by an SMC:  we can define a graphical  \em language \em and  \em calculus \em  \cite{JS, SelingerSurvey} in correspondence to the axioms of an SMC; this procedure traces back to Penrose's work in the 1970s \cite{Penrose}.  
For each morphism formula ${\cal F}$ there is a corresponding \em diagram \em in the graphical language, e.g.~for ${\cal F}=(1_A \otimes k) \circ h\circ (f\otimes g)$ the corresponding diagram is
\[
\epsfig{figure=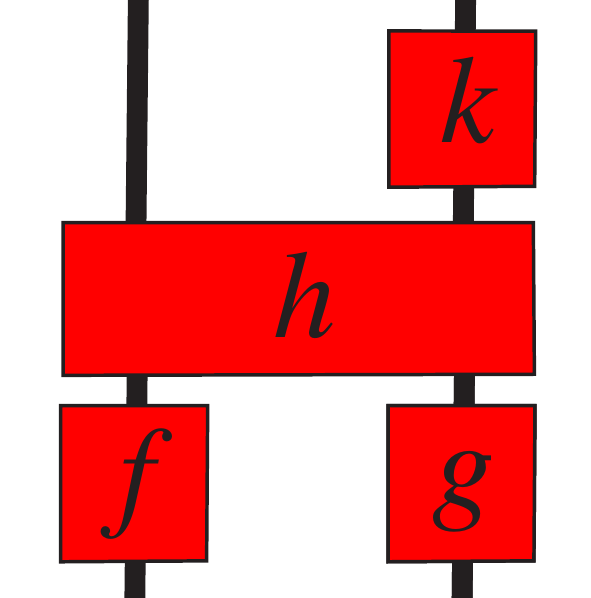,width=58pt}
\]
Surprisingly, a morphism formula still has more information than its corresponding diagram in the graphical  language. 
But from a physical perspective this extra information is in fact redundant.  For example, when expressing both sides of Eq.~(\ref{eq:bifunct}) in the graphical language,  we  obtain the same diagram twice: 
\[
\epsfig{figure=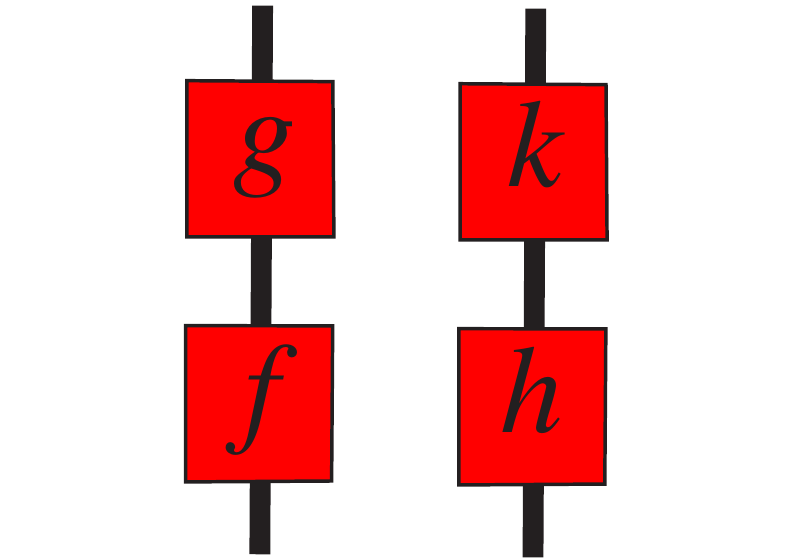,width=70pt}
\]
since  each side represents the same physical `scenario'.  Hence,  the graphical  language is, in a sense, superior to the symbolic language, since it renders an equational constraint superfluous.

The true power of the graphical calculus as opposed to the symbolic  formalism is made clear by the following theorem due to Joyal and Street \cite{JS, SelingerSurvey}, which implicitly defines what we actually mean by `graphical calculus'.


\begin{thm}\label{thm:SMCcoherence}
An equation between morphism formulae in the symbolic language of symmetric monoidal categories follows from the axioms of symmetric monoidal categories if and only if 
we can obtain one picture from the other by  displacing  the boxes, whilst retaining how wires connect the inputs and outputs of  boxes, as well as keeping the overall  number of  inputs and outputs of the diagram fixed.
\end{thm}
So a `graphical calculation' is nothing but a `deformation', for example:
\beq\label{eq:planar_isotopy}
\raisebox{-18mm}{\epsfig{figure=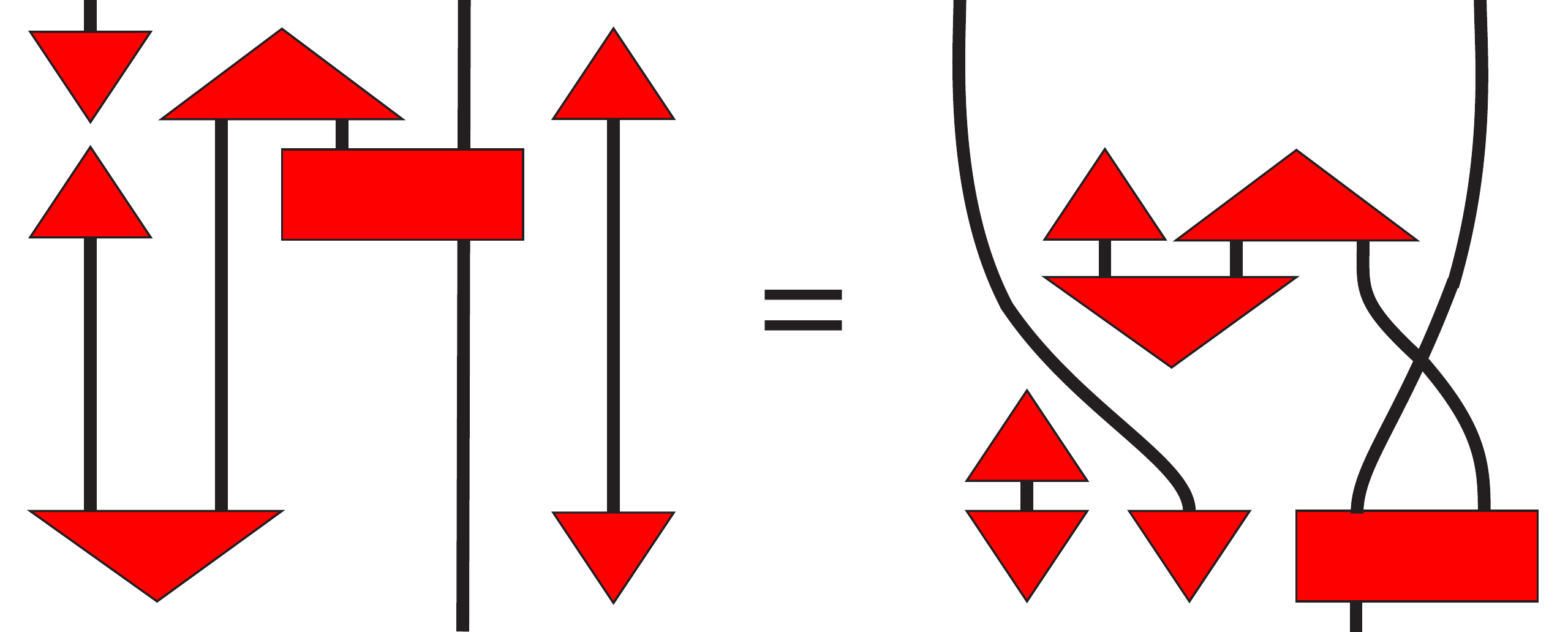,width=250pt}}\vspace{-1.5mm}   
\eeq


From now on we shall make use of the efficiency of the graphical language in making certain equations superfluous: we will treat morphism formula up to equivalence in the diagrammatic representation. 
For example, consider 
the following mathematical ambiguity about our use of connected vs.~disconnected composition as defined \em symbolically\em: while parallel composition \em always \em leads to topological disconnectedness, sequential composition may lead to either a connected or a disconnected diagram. In particular, when we compose over the tensor unit $I$ the two modes of composition coincide:
\[
g\circ f =\  \raisebox{-8.8mm}{\epsfig{figure=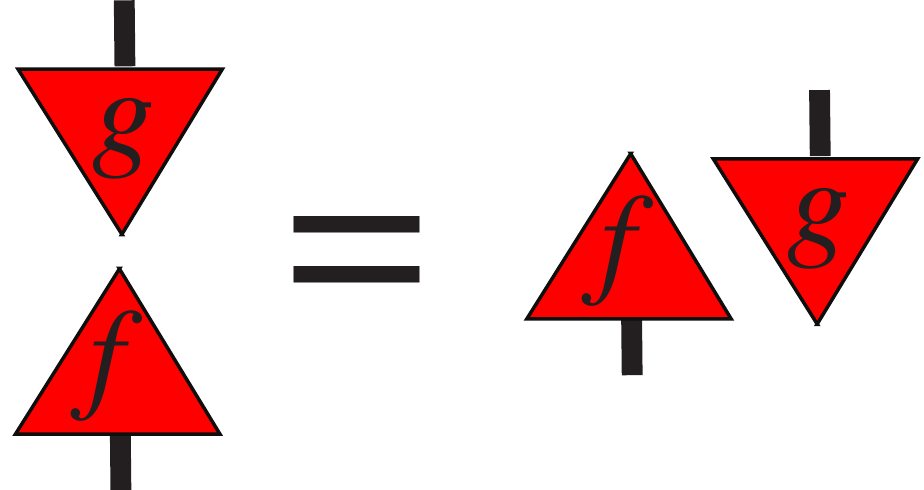,width=80pt}}\ = f\otimes g\,.
\]
Our use of diagrammatic equivalence classes resolves this ambiguity, since we can always represent a composition over the tensor unit by a formula that uses '$\otimes$' instead of '$\circ$'. 

We shall also assume that all our morphism formulae contain only `atomic' expressions---those which do not contain, in the correspondong graphical representation, topologically disconnnected components.  To define this symbolically, we first define  a \em generalized symmetry morphism \em to be a morphism that is either the identity morphism or is the vertical or parallel composition of symmetry morphisms $\sigma_{A,B}:A\otimes B\rightarrow B\otimes A$ or identity morphisms.

\begin{defn}[Non-trivial parallel composition; atomic morphism]\em
The  \em non-trivial parallel composition \em of morphisms $g_1:A_1\to B_1$ and $g_2:A_2\to B_2$  is a morphism   $f=g_1\otimes g_2$, where
neither $g_1$ nor $g_2$ is a scalar (i.e.~of type $I\rightarrow I$). A morphism $f:A\to B$ is \em atomic \em if, for all post or pre--compositions of generalized symmetry morphisms, it cannot be written as a non-trivial parallel composition of other morphisms. 
\end{defn}

Examples of non-atomic morphisms are: 
\[
\raisebox{-5mm}{\epsfig{figure=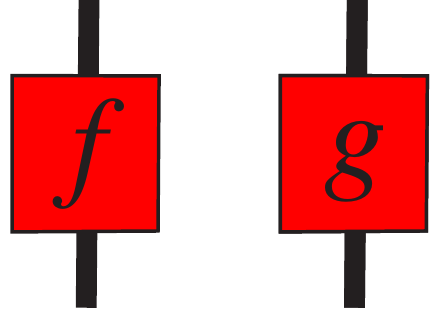,width=40pt}}\ :A_1\otimes A_2\to B_1\otimes B_2
\qquad
\raisebox{-4.5mm}{\epsfig{figure=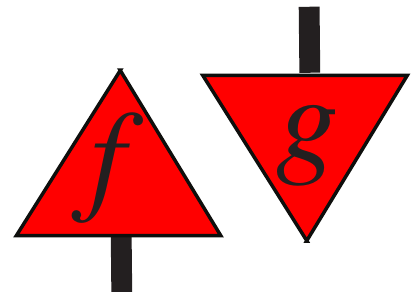,width=40pt}}:A\otimes I\to I\otimes B 
\]
which in the diagrammatic language always consist of two non-trivially typed (i.e.~non-scalar) subcomponents.  
In Appendix A we consider how these assumptions  
affect the relationship between the structure of symbolic formulae and topological connectedness in the corresponding diagrammatic language.
 
\paragraph{Why formulate physical theories using categories?}
A category provides not only a description of objects and morphisms,
but it also provides an \em equational theory\em.  Indeed, in the case
of an SMC, a category provides both the description of a physical
theory,  in terms of a language involving systems, processes
(cf.~evolution) and their (de)composition,  as well as the equational
laws it is subject to, e.g.~when two scenarios or protocols result in
the same overall a process. A particular case of this---which is the
one that one encounters in more traditional formulations of the
dynamics of a theory---is when the final states coincide  given they
take the same input state.  Leaving inputs open then means allowing
for variable inputs.  From a logical perspective this means that it
provides both the \em language\em, i.e.~well-formed formulae (wff),
and the \em axioms\em, i.e.~equations between wff.

By a model one means a concrete realization (e.g.~processes described
in concrete Hilbert-space quantum theory) which typically will obey
some extra axioms and have a more refined language; these can be seen
as additional laws and data, which may or may not be redundant.  For
example, when passing to a theory of quantum gravity one may expect
that certain ingredients of the Hilbert-space structure may have to be relinquished
(e.g.~the continuum \cite{Isham}).

%
%

\subsection{Elements of Categorical Quantum Mechanics}

There have been many attempts to identify the key underlying structures in quantum theory.  For example, the first such attempt, by Birkhoff and von Neumann, used non-distributive lattices to recast quantum theory as quantum logic \cite{vN, BvN}. Other axiomatic frameworks have variously taken as a starting point algebras of observables---using C*-algebras  \cite{Redei, HaagKastler}---or probabilistic structure---using probability spaces \cite{Mackey} or convex structures \cite{Ludwig1, Ludwig2}.   But the focus of these approaches has not usually been on causality.  Relatedly, a weakness of these approaches is that they lack 
an elegant conceptual account of the behavior of \em compound \em quantum systems. Indeed, for most of these approaches the Hilbert space tensor product does not lift to the level of the languages in which the axiomatic framework was stated.  The rise of quantum information and computation, where the tensor product plays a key role, and for which non-local phenomena are exploited for practical applications,  might be seen as  the fatal blow for many of these approaches. 

In contrast, CQM treats composing systems (and processes) as a primitive. This leads to a paradigmatic shift from treating measurement as the basic concept of a theory, as advocated by Birkhoff andvon Neumann, to compoundedness, as advocated by Schr\"odinger \cite{Schrodinger}. This has led to immediate results, and CQM has established itself as a promising framework for studying the foundations of quantum mechanics, as well as a high-level framework for quantum information and computation.  Some milestones and key results of CQM are \cite{SelingerCPM, Vicary, CPav, AbrClone, CD, CES, DP, CPer}.

In CQM we  add expressive power to the formalism described in Subsection \ref{sec:monoidal}  by using \em dagger compact categories\em,  which we define diagrammatically as follows:
\bit
\item {\it dagger}\,: For each graphical element, including an entire diagram, the one obtained by flipping it  upside-down is also a valid graphical element:
\[
\raisebox{-4mm}{\epsfig{figure=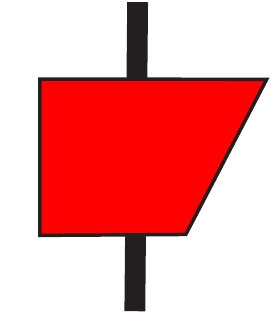,width=26pt}} \leadsto \raisebox{-4mm}{\epsfig{figure=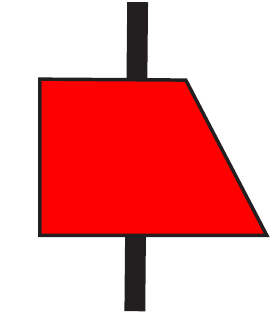,width=26pt}}
\]
\item {\it compactness}\,: for any object $A$ there is an object $A^*$ and a   \em Bell state \em 

\[
\raisebox{-3.5mm}{\epsfig{figure=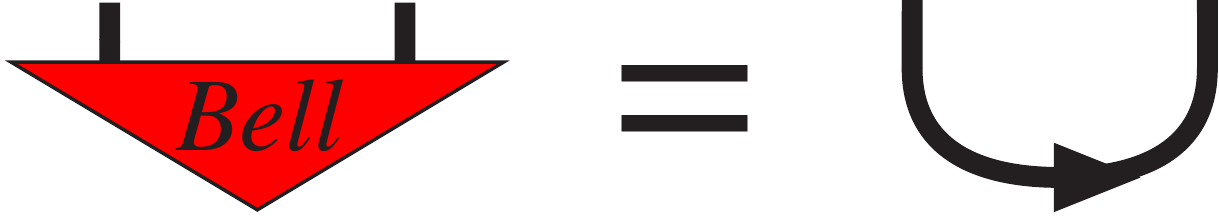,width=118pt}}\
:I\to A\otimes A^{*}\,,
\]
which is such that, 
any two diagrams with matching inputs and outputs are equal if they are equal up to {\bf homotopy}, that is, only the topology of the diagrams matters.   
\eit

In combination with the  dagger structure, this implies:  
\beq\label{pic:tele}
\raisebox{-7mm}{\epsfig{figure=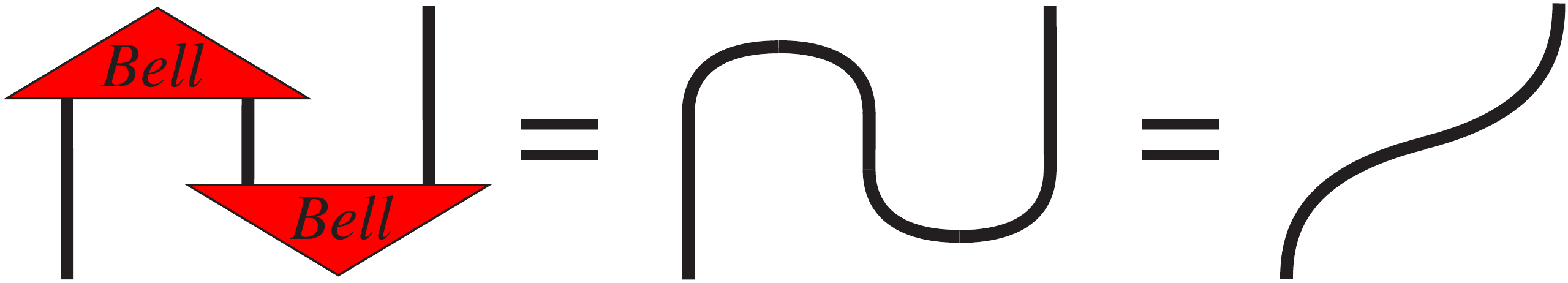,width=240pt}}\vspace{-1.5mm} 
\eeq
\beq\label{pic:teleswap}
\raisebox{-7mm}{\epsfig{figure=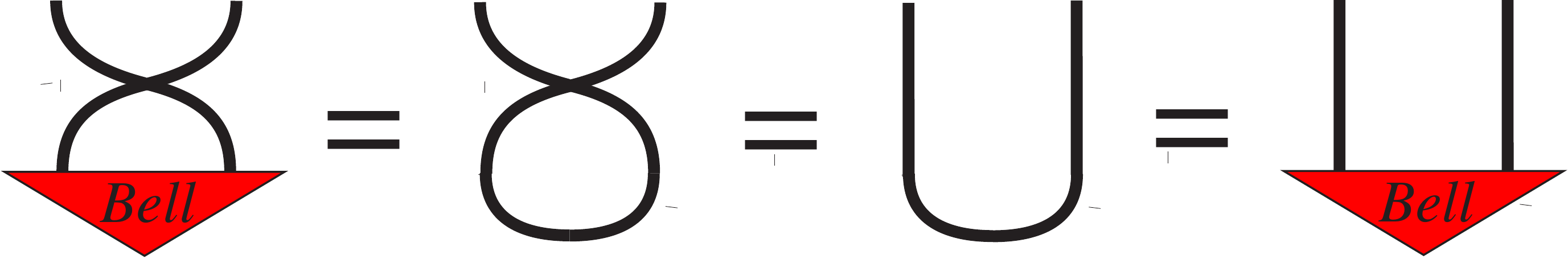,width=258pt}}\vspace{-1.5mm} 
\eeq
These equations are the defining equations of dagger compactness.

\begin{defn}[$\dagger$-SMC]\label{def:compact}\em
A \em dagger category \em is a category equipped with an involutive identity-on-objects contravariant functor $\dagger:\mathbf{C}^{op}\rightarrow \mathbf{C}$, called a \em dagger functor\em. A \em dagger (strict) symmetric monoidal category \em ($\dagger$-SMC) is a (strict) SMC equipped with a dagger functor such that for all  $A,B\in\vert \mathbf{C} \vert$ we have $\sigma^{\dagger}_{A,B}=\sigma^{-1}_{A,B}$, and
\[
(A\otimes B)^{\dagger}= A^{\dagger}\otimes B^{\dagger},
\]
and for all $f:A\rightarrow C,g:B\rightarrow D$ in $\mathbf{C}$:
\[
(f\otimes g)^{\dagger}= f^{\dagger}\otimes g^{\dagger}.
\]
\end{defn}

\begin{defn}[Compact structure]\em
A \emph{compact structure} on an object
$A$ of a SMC $\mathbf{C}$ is a quadruple $(A,A^{*},\epsilon:A\otimes A^{*}\rightarrow I,\eta:I\rightarrow A^{*}\otimes A$)
consisting of $A$, its \emph{dual object $A^{*}$}, the \emph{unit
$\eta_{A}$} and the \em counit \em $\epsilon_{A}$, such that the following
diagrams commute: 
\[
\xymatrix@=2cm{A^{*}\ar[d]_{\eta\otimes A^{*}}\ar[dr]^{1_{A^*}} & A\ar[r]^{A\otimes\eta}\ar[dr]_{1_A} & A\otimes A^{*}\otimes A\ar[d]^{\epsilon\otimes 1_A}\\
A^{*}\otimes A\otimes A^{*}\ar[r]_{1_{A^*}\otimes\epsilon} & A^{*} & A}
\]
A \em compact category \em is one for which there is a compact structure on each object.
\end{defn}

\begin{defn}\label{def:Bell}\em
In a $\dagger$-SMC a \emph{Bell state }$(A,A^{*},\eta)$
is a compact structure $(A,A^{*},\eta^{\dagger}\circ\sigma_{A,A^{*}},\eta)$, and a \emph{dagger compact category} is a\emph{ }$\dagger$-SMC for which
each object has a chosen Bell state;  these choices are moreover coherent with dagger symmetric monoidal structure.
\end{defn}


\begin{example}
The category of finite-dimensional Hilbert spaces and linear maps, denoted $\mathbf{FdHilb}$, is a $\dagger$-SMC for which the dagger functor is given by the linear-algebraic adjoint, and the monoidal product is given by the tensor product of Hilbert spaces. By linearity, states $I\rightarrow \mathcal{H}$ are in bijective correspondence with pure states $|\psi\rangle\in \mathcal{H}$ by the mapping $\psi::1 \mapsto |\psi\rangle$, and the scalars are endomorphisms $\mathbb{C}\rightarrow \mathbb{C}$, i.e. in bijective correspondence with the complex numbers. The terminology of Definition \ref{def:Bell} is justified by the fact that the morphism $\eta::1\mapsto |\eta\rangle$ is given by the quantum state $|\eta\rangle=|00\rangle+|11\rangle$.
\end{example}

We now introduce some operational terminology for the object and morphism languages, with a view to their physical interpretation.

\begin{defn}[Operational terminology]\label{def:slice_etc}\em
A \em slice \em is an object formula  in the symbolic language of SMCs, and a  \em scenario \em or \em protocol \em  is a morphism formula in the symbolic language of SMCs, or equivalently,  a diagram in the graphical language of SMCs.   
\end{defn}

Dagger-compact categorical structure was used in \cite{AC1} to provide sufficient structure to do a large amount of quantum theory, which justifies the terminology of Definition \ref{def:slice_etc}.  They appeared earlier in the work of Baez and Dolan \cite{BaezDolan}, as a particular case of $k$-tuply monoidal $n$-categories; the importance of the particular case of $n=1$ and $k=3$   was later acknowledged by Baez in \cite{Baez}. There also exists an analogous theorem to Theorem \ref{thm:SMCcoherence} for dagger compact categories, which identifies symbolic axioms with a graphical language for which only topology matters \cite{KL, SelingerCPM}.  

 Perhaps even more importantly, the  \em completeness theorem \em by Selinger  \cite{SelingerCompleteness}  states that an equational statement in the language of dagger compact categories is provable in the corresponding language \em if and only if \em it is provable in the SMC of Hilbert spaces, linear maps, the tensor product and the linear-algebraic adjoint. Put informally, for an important class of equational statements, derivability in the graphical language is equivalent to validity within Hilbert-space quantum theory.
Hence a less dichotomic view on
axioms versus models can be obtained by means of the concept of
abstraction. 


\subsection{A pitfall}

As discussed in \cite{ContPhys}, Eq.~(\ref{pic:tele}) can be interpreted as \em post-selected quantum teleportation\em, that is, quantum teleportation conditioned upon the measurement outcomes, such that no unitary correction is needed.
However, naive causal interpretation yields:
\beq\label{telecrash}
\raisebox{-8mm}{\epsfig{figure=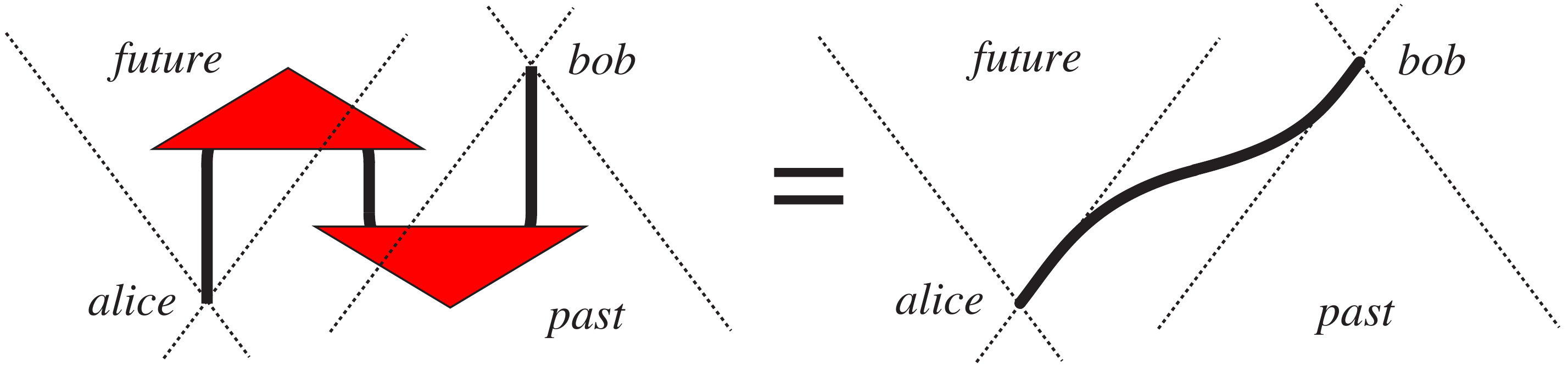,width=280pt}}
\eeq
The origin of this apparent ability for alice to `signal' to bob is  post-selection, which is easily seen to be a (virtual) resource  that enables signaling.    We obtain this even for classical probability theory: if alice and bob each have an unknown bit with the promise that they are the same, then if alice post-selects $x$ then consequently bob will also have $x$. Hence alice has signaled the bit $x$ to bob.  To avoid this,  one must consider all possible measurement outcomes together. In the quantum teleportation protocol this requires classical communication:  
\[
\epsfig{figure=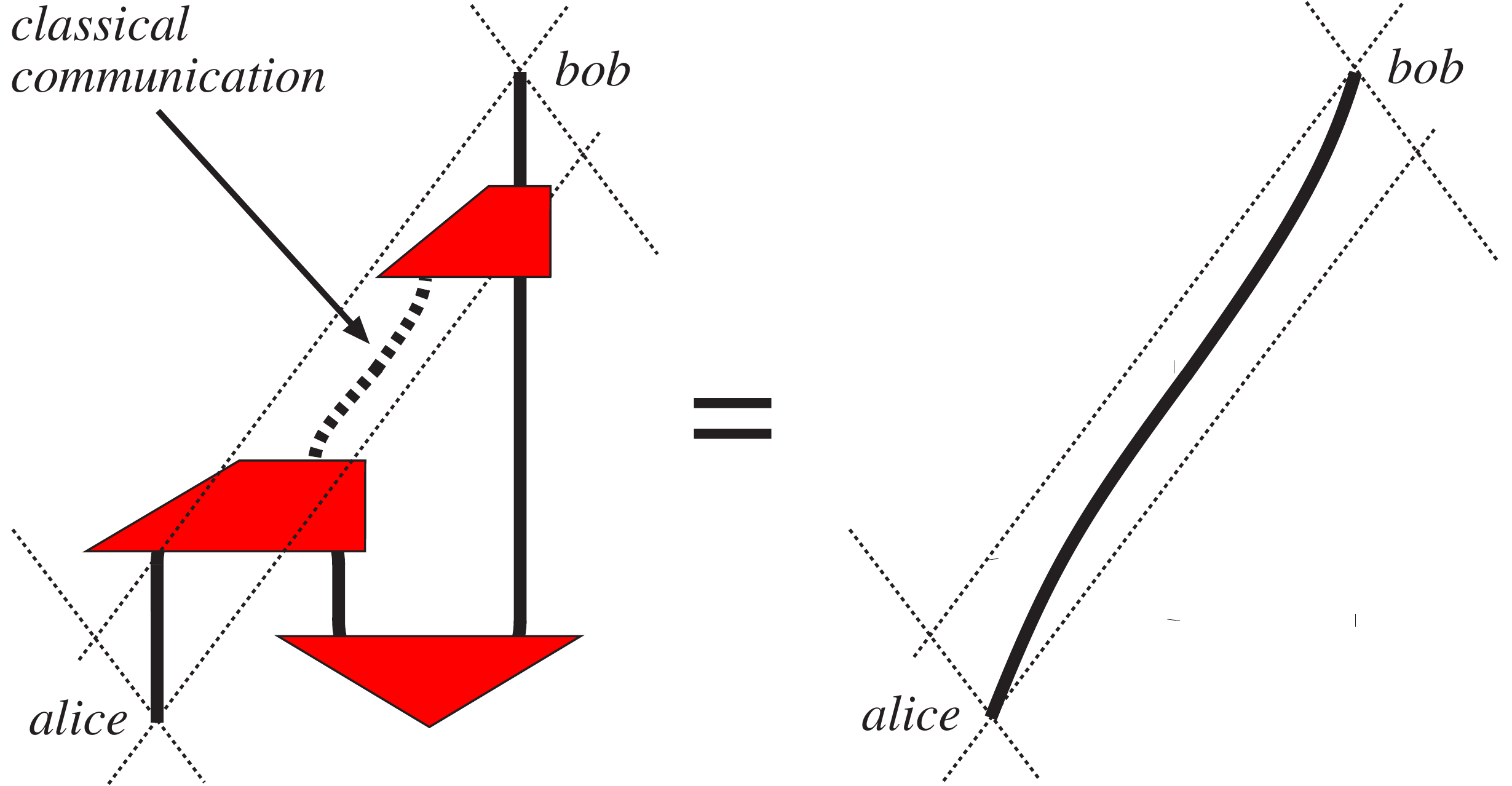,width=240pt}
\]
But to do so in the existing formalism  of CQM requires specifying admissible operations, e.g.~projective measurements, classical communication, classical control structure etc., and to do this various \em internal \em structures must be defined. 
However, in the structure we develop in this paper, \em causal categories\em, postselection will be automatically excluded,
as we see below in Section \ref{sec:restriction2}.  

\section{Terminality of the tensor unit from correlations}\label{sec:Derivation}


In this section we first show how causal structure can be thought of in terms of information flow, and how connectedness captures this. 

\subsection{Causality as information flow, formalised by connectedness}\label{sec:CausalityInSMC}

Causal structure is often conceived as a partially ordered set $({\cal A}, \leq)$ where $A\leq B$ stands for $A$ being in the causal past of $B$. The passage to SMCs will involve more than just expressing that there \emph{is} a causal connection:  it will involve specifying the processes that establish this causal connection from $A$ to $B$, e.g.~by means of sending a non-void signal. 

Now, consider a physical scenario of the kind we discussed in the previous section:
\[
\epsfig{figure=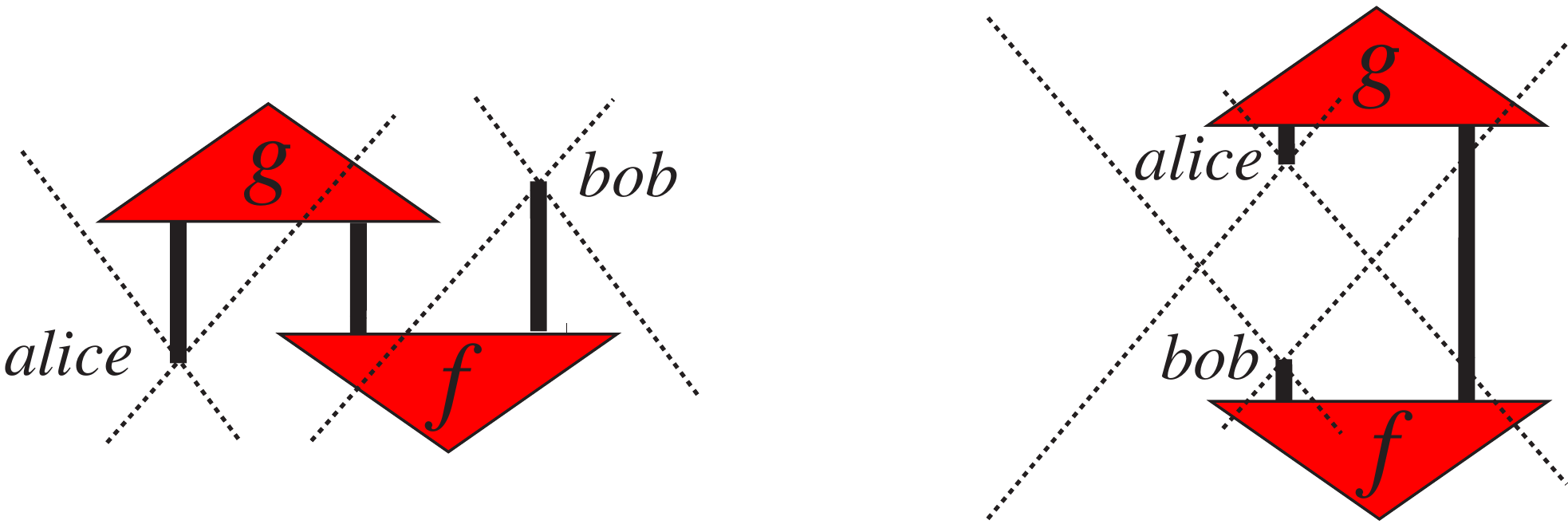,width=200pt}
\]
where  $A$ (alice) is \em not \em causally preceded by $B$ (bob).  Now, whilst by causality no information can flow from $A$ to $B$ (cf.~Eq.~(\ref{telecrash}) and the discussion in the previous section), there does \em physically exist  \em a composite process, e.g.~for the picture on the left: 
\[
h= (g\otimes 1)\circ(1\otimes f) :A\to B\,,
\]
So in particular,
 ${\bf C}(A,B)\not=\emptyset$\,.    Hence we make a  key distinction is between: 
\bit
\item  the \em existence \em of a physical process, that is ${\bf C}(A,B)\not=\emptyset$; and,
\item the \em flow  of information \em enabled by such a process: it is whether information flow  can occur which in this paper will witnesses  a causality assertion $A \leq B$.
\eit

\begin{example}[Proof theory]  The passage from an ordered structure to a categorical structure, or one from \em assertion \em to \em witnessing\em, is exactly what has occurred in logic, specifically in proof theory.  While in algebraic logic one asserts that there \emph{is} a proof which derives  predicate $B$ from predicate $A$, in categorical logic one also articulates \emph{how} this can be established by explicitly  giving the proofs, a proof then being a morphism in some category  of type $A\to B$ (see e.g.~\cite{LambekScott, AbrTze}). So rather than focussing on provability one also takes the structure of the  \emph{space}  of proofs into account:
\[
{\mbox{partial order}
\over
\mbox{category}}
=
{\mbox{asserting}
\over
\mbox{providing witnesses}}
=
{\mbox{algebraic logic}
\over
\mbox{categorical logic}}
=
{\mbox{causal structure}
\over
\mbox{``content of this paper''}}
\]
But in proof theory the paradigm connecting the ordered structure and the categorical structure is: 
\[
A \leq B \ \Longleftrightarrow\ {\bf C}(A,B)\not=\emptyset\,.
\]
or in words, $B$ is derivable from $A$ if there exists a proof that does so.  The above discussion shows that this proof-theory paradigm cannot be retained on-the-nose, and rather than having existence of a morphism as 
witnessing partial ordering, we will require the existence of an information-flow-enabling morphism.
\end{example}

We shall now formalize what we mean by information flow. 

\begin{remark}
In what follows we have categories of \em deterministic processes \em in mind, i.e.~non-postselected.  In CQM terms, this means that the category only contains one scalar, namely $1_{I}$, representing `certainty'.
\end{remark}

\begin{defn}\label{def:pointed1}\em 
We say that a process $f:A\to B$ is: 
\bit
\item \em constant on states \em iff for all $\psi, \phi:{I}\to A$  we have $f\circ\psi=f\circ\phi$;
\item \em is determined by its action on states \em iff for all $g:A\to B$,
$f\circ\psi=g\circ\psi$ for all $\psi:{I}\to A$ implies $f=g$.
\eit
\end{defn}

In this paper \em information flow \em means a non-constant process $f:A\to B$. Since in this case there exists $\psi, \phi:{I}\to A$ with $f\circ\psi\not=f\circ\phi$,  in a scenario bob can choose to `feed' either $\psi$ or $\phi$ into $f$, so that alice receives  $f\circ\psi$ or $f\circ\phi$ respectively: 
\beqa
&&\hspace{-2.1cm}\underbrace{\epsfig{figure=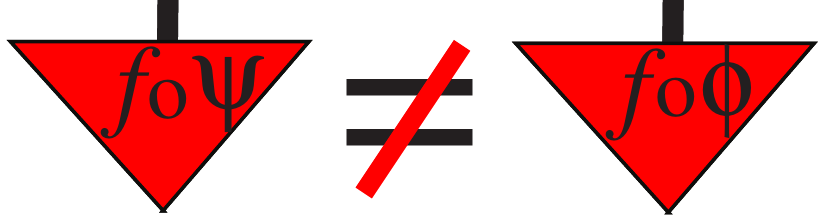,width=65pt}}\\
&\epsfig{figure=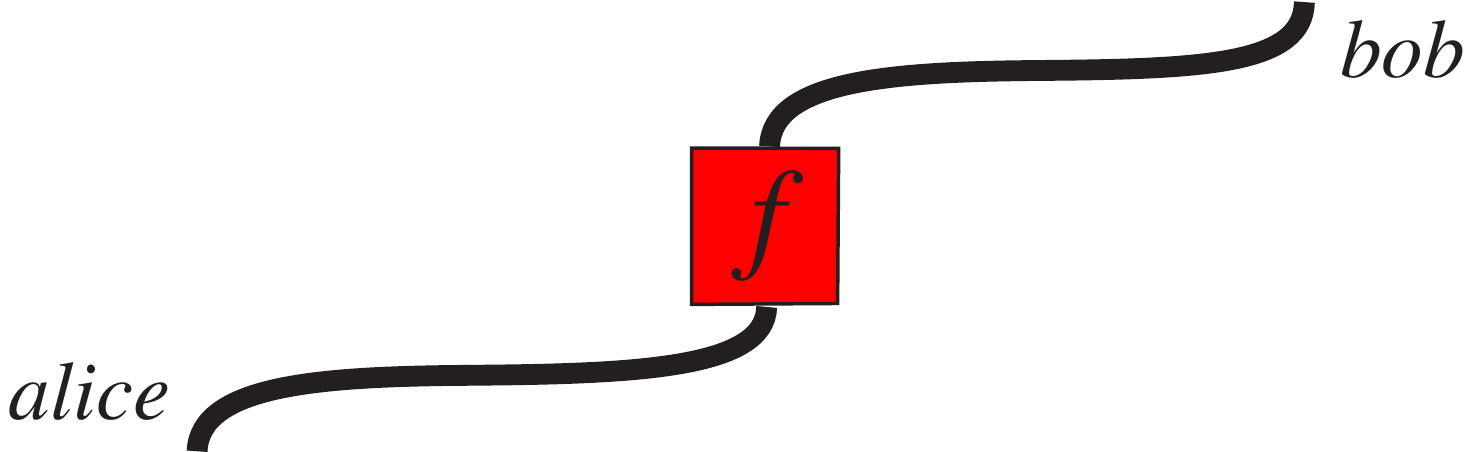,width=150pt}&\\
\overbrace{\epsfig{figure=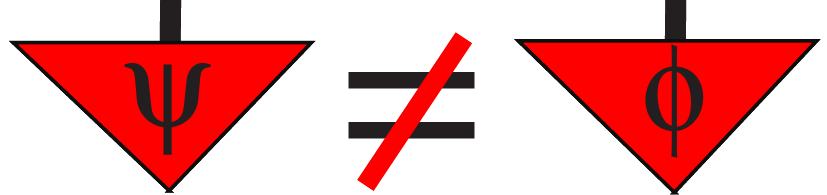,width=65pt}}\hspace{-2.2cm}&&
\eeqa

\begin{remark}[Well-pointedness]
Conditions of `well-pointedness', as used in Definition \ref{def:pointed1}, are sometimes thought to be undesirable, both for mathematical and physical reasons \cite{Johnstone, Isham}.  However, our level of generality will also capture `pointless' objects: we will show that, in a well-pointed situation, the notion of information flow that refers to points can be equivalently stated purely in terms of connectedness, without reference to states.  It is then this pointless characterization that we will use throughout the paper.  
\end{remark}

We shall now establish that  information flow  from $A$ to $B$ is captured by topological  connectedness in the graphical language:
\begin{center}
\begin{tabular}{|c|c|} 
 information flow  & no information flow \\ 
\hline
\epsfig{figure=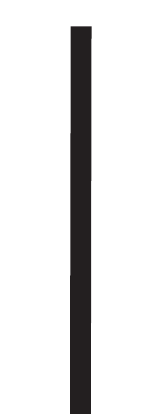,width=16pt}   & \epsfig{figure=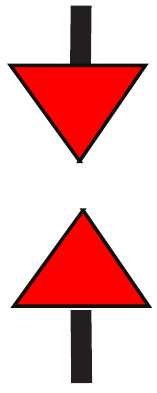,width=16pt} \\
\hline
\end{tabular} 
\end{center}


\begin{defn}\label{def:disconnected}\em
In an SMC,  a morphism  $f:A\to B$ is  \em disconnected \em if it decomposes as $f=p\circ e$ for some $e:A\to I$ and $p:I\to A$. A morphism is \em connected \em if it is not disconnected.
\end{defn}



\begin{prop}[Equivalence of constancy and disconnectedness]\label{propo:infoflowpoints}
If all scalars are equal to $1_{I}$ and if processes are determined by their action on states, then the following are equivalent:
\bit
\item $f:A\to B$ is constant on states;
\item $f:A\to B$ is disconnected. 
\eit 
\end{prop}
\begin{proof}
Let $f$ be constant on states and $\phi=f\circ \psi$ be that constant. Then we indeed have $f=\phi\circ\pi$ for any $\pi:A\to{I}$, since: 
\[
(\phi\circ\pi)\circ\psi=\phi\circ\underbrace{(\pi\circ\psi)}_{1_{I}}=\phi=f\circ \psi\vspace{-4.5mm}
\]
for all $\psi:{I}\to A$. Conversely:
\[
(\phi\circ \pi)\circ\phi = \phi\circ(\pi\circ\phi)=\phi\circ(\pi\circ\phi')=(\phi\circ \pi)\circ\phi'
\]
 where we again used uniqueness of scalars.
\end{proof}

 Hence we have characterized information flow using the structure of an SMC: $A$ and $B$ \em are causally related \em if a process $f:A\to B$ can take place which is not disconnected, and dually, $A$ and $B$ \em are not causally related \em iff all processes $f:A\to B$ are disconnected (i.e.~factor through the tensor unit), that is:
\[
{\bf C}(A, B)=\{\psi\circ\top_A \mid \psi:{I}\to  B\}
=\left\{\left.\raisebox{-4.5mm}{\epsfig{figure=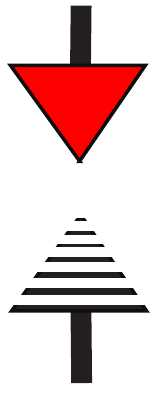,width=12pt}}\ \ \right| \
\raisebox{-1.5mm}{\epsfig{figure=CCstate.pdf,width=12pt}}:{I}\to  B \right\} 
\]
where $\raisebox{-1.5mm}{\epsfig{figure=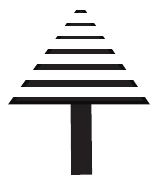,width=12pt}}: A\to{I}$ is the unique effect, which is stated in anticipation of 
the main result of the following section.

\subsection{Terminality of ${I}$ as `no  correlation-induced signaling'}\label{sec:restriction2}

Our notion of causality has been based so far on information flow between distinct locations. In the previous Subsection this was enabled by a process $f:A\to B$. We shall call this \em information flow of the 1\st kind\em.  However, given a bipartite state $f:{I}\to A\otimes B$, there may also be another type of information flow, which we call \em information flow of the 2\nd  kind\em. Diagrammatically, they appear as follows:

 \begin{center}
\begin{tabular}{|c|c|} 
1\st kind info-flow & 2\nd kind info-flow \\
\hline
\begin{minipage}{2.4cm}
\vspace{2mm}
\begin{center}
$\ \ \underbrace{\epsfig{figure=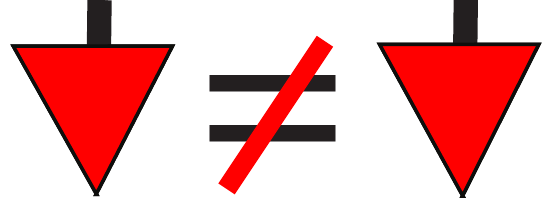,width=35pt}}\ \  $\vspace{1.5mm}\\
$\epsfig{figure=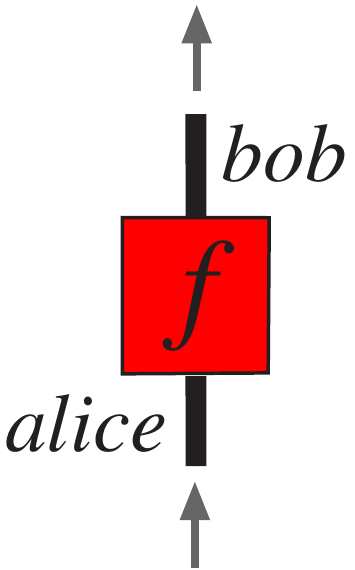,width=40pt}\ \  $\vspace{1.5mm}\\
$\  \ \overbrace{\epsfig{figure=psismall.pdf,width=35pt}}\ \  $\vspace{2mm} \\ 
\end{center}
\end{minipage}
& 
\begin{minipage}{4.0cm}
\begin{center}
$\underbrace{\epsfig{figure=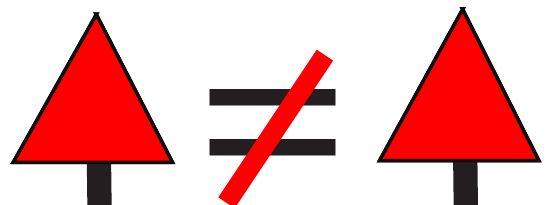,width=35pt}}\ \ \underbrace{\epsfig{figure=psismall.pdf,width=35pt}}\!\!\!$\vspace{1.5mm}\\
$\epsfig{figure=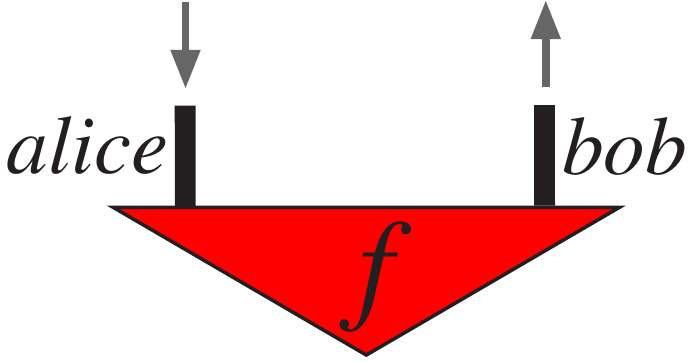,width=80pt}$\\ 
\end{center}
\end{minipage}\\
\hline
\end{tabular} 
\end{center}

\begin{example}[Quantum entanglement]
In quantum theory a bipartite state may be entangled. In that case, information flow of the 2\nd kind corresponds to correlations between measurement outcomes of the two parties that are \em signaling\em, which can only happen if we allow post-selection \cite{MRV}.
\end{example}

The following postulate imposes compatibility between information flows of the 2\nd kind  with those of the 1\st kind; in other words, it forbids correlation-induced signaling when systems are not causally related: otherwise 2\nd kind information-flow could be used to produce 1st kind information-flow, thus violating causal structure. 

\begin{post}[Causal consistency]\label{epostulate}
For a bipartite state $f:{I}\to A\otimes B$, information flows of the 2\nd kind cannot  occur when $A$ and $B$ are not causally related.
\end{post}

\begin{remark}
Note that we could have made the stronger requirement that entanglement-induced signaling does not occur even for causally related systems, but since these `information flows of the 2\nd kind across time' do not cause any inconsistency with causal structure we ignore them.
\end{remark}

If all bipartite states are \em disconnected\em:
\begin{center}
\epsfig{figure=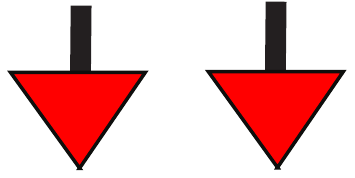,width=36pt}\vspace{2mm}
\end{center}
then, by analogy with the disconnected processes that characterize causal independence, 
there will  be no information flow of the 2\nd kind, hence Postulate \ref{epostulate} is trivially satisfied.
However, the kind of universes that interest us of course do have connected bipartite states,
both in quantum theory and for classical probabilistic states, for example, a perfectly correlated bipartite state.  The following definition asserts the existence of states of this kind: it states that all processes can be faithfully represented by bipartite states. In the context of quantum theory this corresponds to the Choi-Jamio\l{}kowski isomorphism \cite{Choi,Jam}, as described in Example \ref{ex:mixedQT} below. Technically, it weakens the definition of compactness  (see above) which, as we shall in Subsection \ref{sec:Incompatibilities}, cannot be retained.  

\begin{defn}\label{Def:CJuniv}\em
In a \em CJ-universe \em for all systems $A$ there exists another system $A^*$, not causally related to $A$, as well as a bipartite state $\eta_A:{I}\to A^*\otimes A$,  called the \em CJ-state\em, such that for all $B$, 
\[
{\bf C}(A, B)\to {\bf C}({I}, A^*\otimes B)::f\mapsto (1_{A^*}\otimes f)\circ\eta_A
\]
is an injective mapping. 
\end{defn}

Note that this definition implies in particular that for the case $B={I}$, there is an injection from effects $\pi:B\to {I}$ to states $(1_{A^*}\otimes \pi)\circ\eta_A: {I}\to A^*$.

\begin{defn}\label{Def:environment_structure}\em
By an \em environment structure \em we mean a family of effects $\top_A:A\to{I}$, one for each system $A$.  
We call a CJ-universe with an environment structure a \em  CJ$\top$-universe\em.
\end{defn}

We can interpret these processes as `removing that system from our scope'.  In quantum theory this role is played by the partial trace operation.   



\begin{defn}\em
A \em terminal object \em in a category $\mathbf{C}$ is an object $A$ for which, for each object $B\in | \mathbf{C}|$, there is a unique morphism from $B$ to $A$.
\end{defn}

Note that the uniqueness of $1_{I}$ as a scalar is implied by terminality of the tensor unit.

\begin{thm}\label{thm:CJuniv}
A CJ$\top$-universe obeying Postulate \ref{epostulate} has a terminal tensor unit.
\end{thm}
\begin{proof}
If $\pi\not=\pi': A\to{I}$, then by Definition~\ref{Def:CJuniv} (with $B:= {I}$) we have 
\[
(1_{A^*}\otimes \pi)\circ\eta_A\not=(1_{A^*}\otimes \pi')\circ\eta_A\,, 
\]
which contradicts Postulate \ref{epostulate}.  Hence there can at most be one effect and its existence is guaranteed by the  environment structure.
\end{proof}

\begin{example}[Classical probability theory]
We define classical probability theory as a subcategory $\mathbf{Stoch}$ of the category of real matrices $\mathbf{Mat}(\mathbb{R})$: morphisms are stochastic maps, i.e. finite-dimensional real matrices with entries $p_{ij}>0$, and whose columns are normalised, i.e.~$\forall j\;\Sigma_i p_{ij}=1$. The monoidal product is the Kronecker product of matrices. States are given by normalised positive-real row vectors, and the environment structure is given by marginalisation of the probability distribution. A CJ-state is then given by a perfectly correlated bipartite probability distribution 
\[
\mathbf{v}_{ij}= \left\{ \begin{array}{ll}
        \frac{1}{n} & \textrm{if } i=j;\\
       0 & \textrm{otherwise}.\end{array} \right.
\]
which can easily be seen to provide an injective mapping from operations to states.
\end{example}

\begin{example}[Mixed quantum theory]\label{ex:mixedQT}
$\mathbf{FdHilb}$ is the motivating example of a $\dagger$-SMC in CQM, and one might attempt to define it as a $CJ\top$-universe, using the Bell state as the CJ state. However this is problematic 
for the following reason. The environment structure provides a unique morphism $\top_A:A\rightarrow I$ for each object $A$, and the interpretation of this family of morphisms $\{\top_A\}_A$ is the partial trace operation (i.e. the operation which sends the system $A$ to the environment). But since tracing out a system in quantum theory  typically leads  to a mixed state when starting with a global pure state (e.g.~a maximally entangled state), we should consider the category of mixed operations rather than $\mathbf{FdHilb}$, whose states are always pure (see also Remark 7 of \cite{CPer}). 

Hence we define a category $\mathbf{Mix}$ whose objects are finite-dimensional Hilbert spaces (i.e. the same objects as $\mathbf{FdHilb}$), and whose morphisms are completely positive maps for the appropriate domain and codomain: denoting the set of linear operators
on $\mathcal{H}$ 
as $L(\mathcal{H})$
we define
\[
\mathbf{Mix}(\mathcal{H}_1,\mathcal{H}_2):=\{f:L(\mathcal{H}_1)\rightarrow L(\mathcal{H}_2)\;|\; f  \textrm{ is completely positive}\}.
\]
Monoidal structure is again given by the tensor product of Hilbert spaces, and the environment structure for $\mathbf{Mix}$ is given by the partial trace. 

Now, we define $|\mathcal{B}\rangle:=\Sigma_i |i\rangle\otimes | i\rangle$ for a fixed orthonormal basis of  $\mathcal{H}_1\otimes\mathcal{H}_2$ (i.e.~the maximally entangled state for $\mathcal{H}_1\otimes\mathcal{H}_2$). Then the CJ state for $\mathbf{Mix}$ is the operator $|\mathcal{B}\rangle\langle \mathcal{B}|$, since it supports the Choi-Jamio\l{}kowski isomorphism $\phi$ from completely positive maps $f:L(\mathcal{H}_1)\rightarrow L(\mathcal{H}_2)$ to positive operators $M$ on $L(\mathcal{H}_1\otimes\mathcal{H}_2)$, given by
\[
\phi:: f\mapsto  (f\otimes 1_{L(\mathcal{H}_2)})\circ|\mathcal{B}\rangle\langle \mathcal{B}|
\]

and whose inverse is
\[
\phi^{-1}::M\mapsto \mathrm{Tr}_{\mathcal{H}_1}[(1_{\mathcal{H}_2}\otimes(-)^T)M].
\]

If we restrict to the subcategory $\mathbf{Mix}_\top$ of $\mathbf{Mix}$ whose morphisms are completely-positive \em trace-preserving \em maps then we obtain a $CJ\top$-universe, i.e.~the tensor unit $\mathbb{C}$ is terminal. Note that the construction of a category of mixed states and operations has been axiomatised in \cite{SelingerCPM}, where the \em CPM construction \em was defined for any $\dagger$-compact SMC.
\end{example}

This has the following trivial consequences.

\begin{cor}
Under the assumptions of Theorem \ref{thm:CJuniv} we have: 
\bit
\item All scalars are equal to $1_{I}$\,;
\item States are `normalized' i.e.~$\top_A\circ\psi=1_{I}$ for all $\psi:{I}\to A$\,;
\item For all $A$, $B$ we have $\top_{A\otimes B}=\top_A\otimes \top_B$\,;
\item All bipartite effects are disconnected.
\eit
\end{cor}

Within this framework we can now show that teleportation without classical communication cannot generate any information flow.

\begin{cor}\label{cor:notele}
Under the assumptions of Theorem \ref{thm:CJuniv}, the composite of the protocol
\[
\raisebox{0mm}{\epsfig{figure=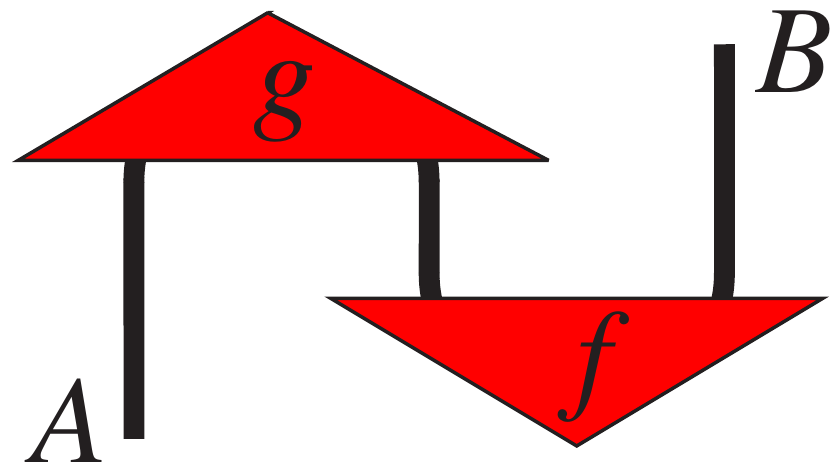,width=80pt}}
\]
must be disconnected, that is, it is of the form:
\[
\raisebox{-7mm}{\epsfig{figure=CC1.pdf,width=16pt}}: \ 
A\to {I}  \to  B
\]
\end{cor}

The diagrammatic proof is as follows:
\[
\epsfig{figure=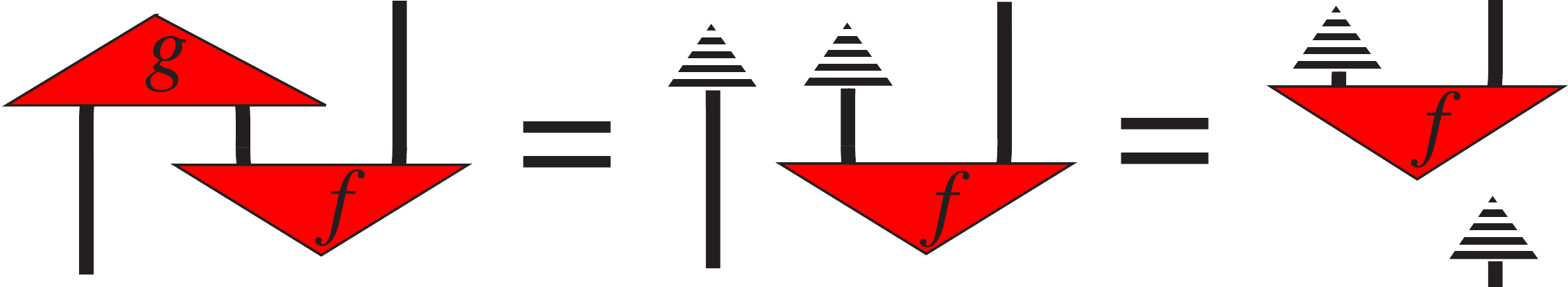,width=246pt}
\]
Similarly we have:
\[
\epsfig{figure=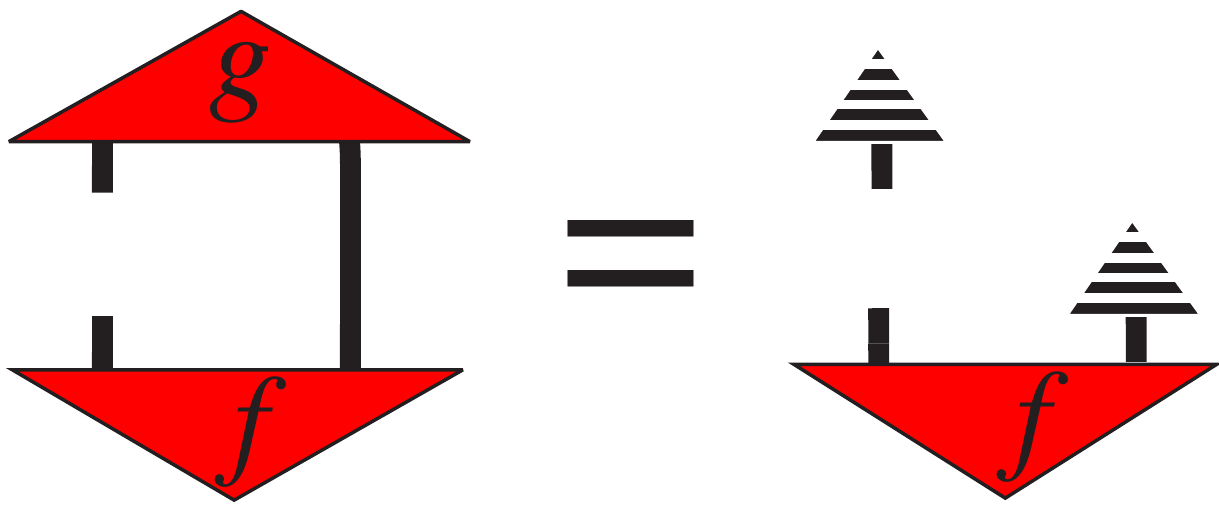,width=124pt}
\]

\begin{cor}\label{cor:disc}
Under the assumptions of Theorem \ref{thm:CJuniv}, in a compact category all identities must be disconnected and hence trivial, and in a dagger category all bipartite states must be disconnected, and each object has at most one state.
\end{cor}
\begin{proof}
The first part is a  direct mathematical analogue to Corollary \ref{cor:notele} and the second part is straightforward.
\end{proof}

A formal account of this is in Section \ref{sec:Incompatibilities}; in Section \ref{sec:constructCC} we show how we can retain the full power of CQM.
 
 \begin{remark}[Time-symmetric quantum mechanics]
The passage from dagger compact categories to causal categories can also be seen as an abstract counterpart to the passage from time-symmetric quantum mechanics (TSQM) \cite{ABL} to the usual formalism of quantum mechanics.   In the formalism of TSQM, not only do we assume the existence of a pre-selected state (i.e.~one which has been prepared by measurement), but we also assume that measurement outcomes have been post-selected. 
This corresponds to how a dagger symmetric monoidal category  is used in CQM, because in this setting the dagger imposes a formal symmetry between states and effects.
In TSQM, the violation of Postulate \ref{epostulate} has been partially addressed by restricting the formalism to those classes of intial and final (post-selected) states which do not lead to signaling \cite{MRV}; this is ad hoc, and these classes lack an elegant formal characterisation.  
\end{remark}

\begin{remark}
In fact, when exploiting the input-output duality at both at alice's and bob's sites we can identify  two more kinds of information flow:
 \begin{center}
\begin{tabular}{|c|c|}
3\rd kind info-flow & 4\rth kind info-flow \\
\hline
\begin{minipage}{2.4cm}
\vspace{2mm}
\begin{center}
$\ \ \underbrace{\epsfig{figure=pismall.pdf,width=35pt}}\ \ $\vspace{1.5mm}\\
$\epsfig{figure=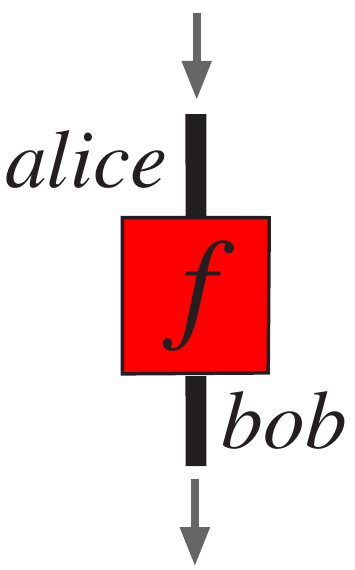,width=40pt}\ \  $\vspace{1.5mm}\\
$\  \ \overbrace{\epsfig{figure=pismall.pdf,width=35pt}}\ \  $\vspace{2mm} \\
\end{center}
\end{minipage}
&
\begin{minipage}{4.0cm}
\begin{center}
$\epsfig{figure=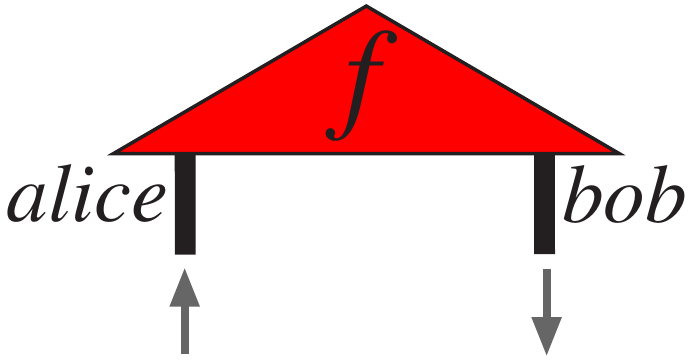,width=80pt}$\vspace{1.5mm}\\
$\overbrace{\epsfig{figure=psismall.pdf,width=35pt}}\ \ \overbrace{\epsfig{figure=pismall.pdf,width=35pt}}\!\!\!$\\
\end{center}
\end{minipage}\\
\hline
\end{tabular}
\end{center}
Each of these is however excluded by terminality; in the first case since there cannot be two non-equal effects and in the second case since the bipartite effect must itself be $\top_A\otimes \top_B$.
\end{remark}

\paragraph{Earlier work: Chiribella, D'Ariano and Perimotti.}   In \cite{CDP1}, Chiribella, D'Ariano and Perimotti use the existence of a unique deterministic effect, which they call the \em causality axiom \em \cite[Definition 25 \& Lemma 3]{CDP1}, to derive information-theoretic features of quantum theory. In that work, much use is made of the probabilistic structure of measurements and classical outcomes. However, in our framework, we can already derive such features without assuming probabilistic structure, e.g.~Corollary~\ref{cor:notele}. By using compositionality, we expose the structural---as opposed to probabilistic---aspects of information flow that follow from requiring causality.

\section{Partiality of the tensor from global state}\label{sec:partial}

In a monoidal category the tensor product exists for every pair of objects, in particular, a system $A$ can be tensored with itself to produce $A\otimes A$.  In that case, ${\bf C}({I}, A\otimes A)$ consists of all $\psi\otimes\phi$ for $\psi, \phi\in{\bf C}({I}, A)$. The independence of $\psi$ and $\phi$  lacks physical meaning, since if the two $A$s in $A\otimes A$ refer to one and the same system (spatiotemporally) then obviously we have $\psi=\phi$.  More generally, the freedom of composing arbitrary states of a system $A$ and a system $B$ implies that they should be independent in $A\otimes B$.  The contrapositive to this is that for systems that are not independent, we have to \em restrict composition of the states of subsystems\em, which can also be achieved simply by constraining the composition of the objects themselves.  We will derive this feature in this Subsection, where we first investigate to which slices (i.e.~object formulae) we can meaningfully assign states within the context of arbitrary scenarios or protocols.  In this manner we will develop a partial tensor product, which will be the analogue of spacelike hypersurfaces in relativity: we call these objects `spatial slices'.

Now, we will assume that in a scenario or protocol each object occurs only once as an input type and once as an output type.  This can be achieved without loss of generality simply by annotating an object with its `location' within the  scenario or protocol. The reason for this assumption is merely to guarantee that in the following definition an object is considered  only once. Recall that we assume that morphism formulae contain only atomic morphisms.

\begin{defn}\em
Let ${\cal F}$ be a morphism formula.   
A slice ${\cal B}$  \em is included in ${\cal F}$ \em if the objects occurring in ${\cal B}$ is a subset of the input and output types of the atomic morphisms that are in $\mathcal{F}$.
\end{defn}

In terms of the graphical language, a slice is included in a diagram just when it is a subset of the wires in it, which we can denote by putting ticks on the wires:
\beq\label{pic:slice}
\raisebox{-9mm}{\epsfig{figure=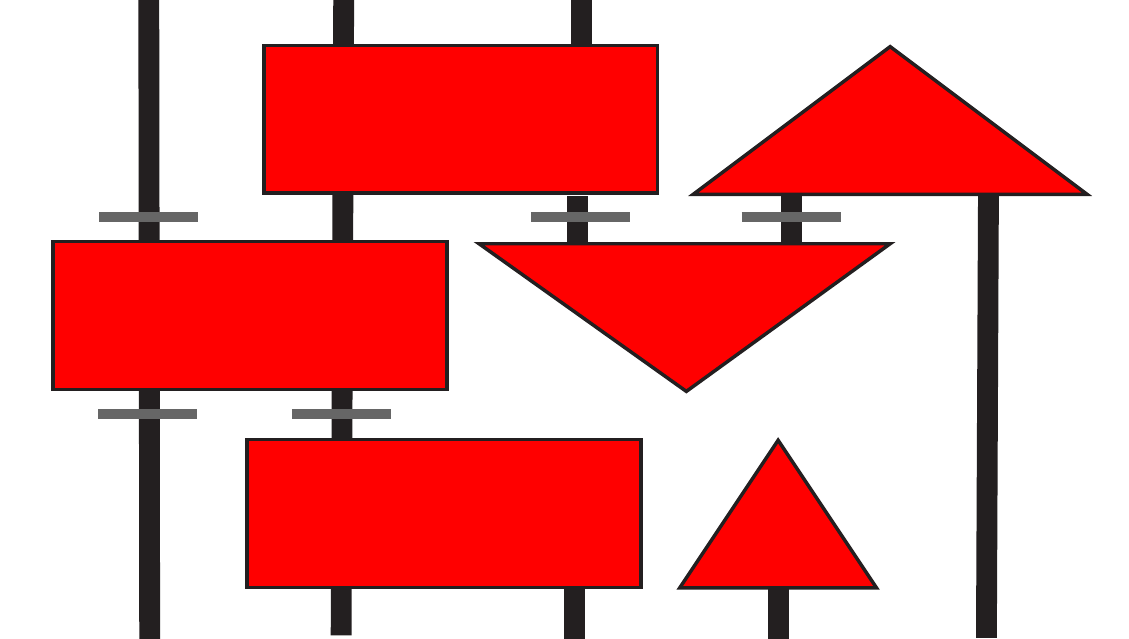,width=115pt}} 
\eeq


\begin{defn}\label{def:disconnectedbis}\em
A \em spatial slice \em is a  slice  $B_1\otimes \ldots\otimes B_n$ for which every pair of objects $(B_i, B_{j\not = i})$ that occur  in it  are \em disconnected objects\em, that is, if $\mathbf{C}(B_i,B_j)$ and $\mathbf{C}(B_j,B_i)$ are disconnected.
\end{defn}

While the slice in picture (\ref{pic:slice}) cannot be spatial, since it involves objects that are explicitly connected within the diagram, the following slice may be spatial:
\beq\label{pic:sliceII}
\raisebox{-9mm}{\epsfig{figure=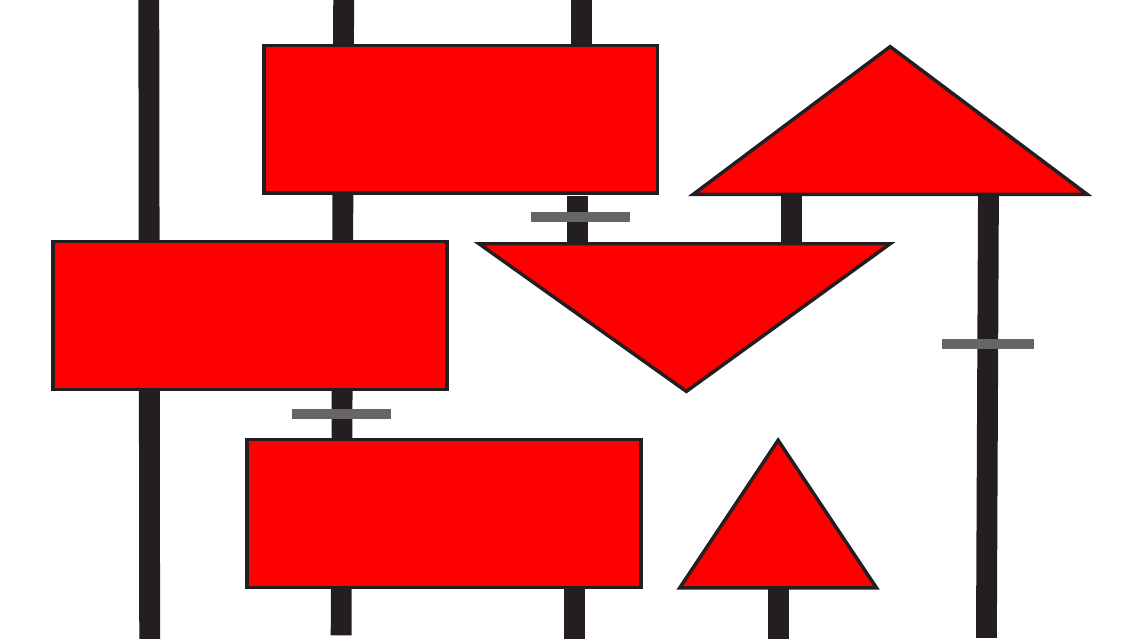,width=115pt}} 
\eeq
provided there are no scenarios for which the ticked objects are connected.

\begin{defn}\em
Let ${\cal F}$ be a morphism formula made up of atomic morphisms.  Another  morphism formula ${\cal G}$ is a \em sub-formula \em of ${\cal F}$, if ${\cal F}$ can be formed from ${\cal G}$ (necessarily included!), $\otimes$, $\circ$ and other morphism formula.  Similarly we define \em sub-scenario \em and \em sub-protocol\em.
\end{defn}

In the graphical language a sub-formula is simply part of a diagram:
\[
\raisebox{0mm}{\epsfig{figure=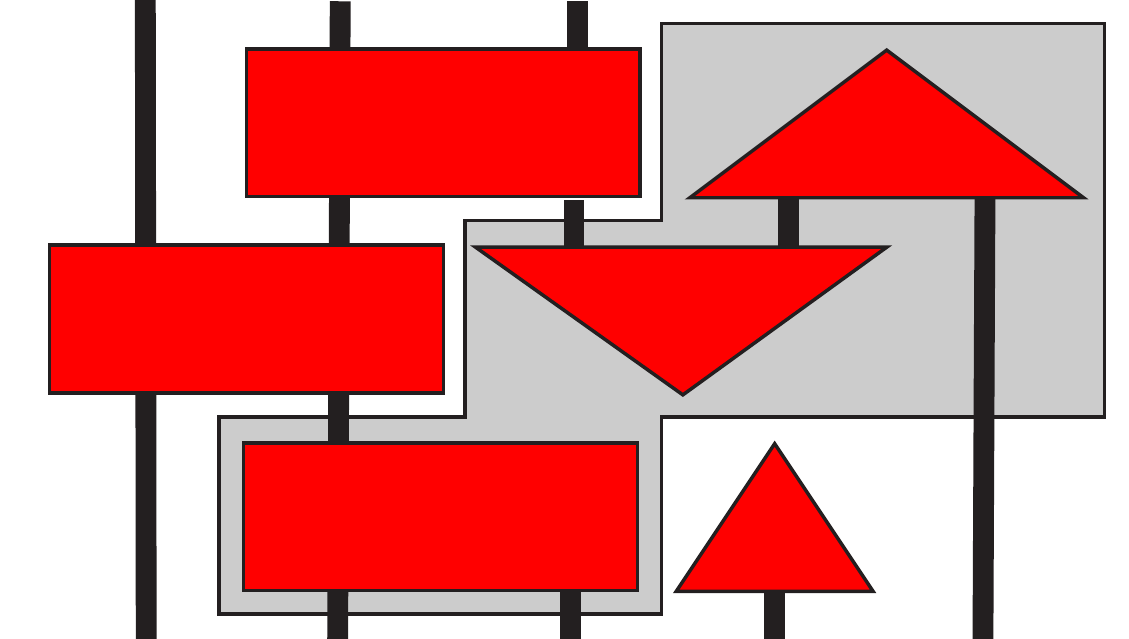,width=115pt}} 
\]
The following definition states the conditions under which we can assign a state to a slice included within a protocol or scenario, when an initial state is specified (cf.~`initial condition'). We denote by $\sigma$ the appropriate  composite of symmetry isomorphisms that realizes the stated type.

\begin{defn}\label{def:localstate}\em
Let ${\cal F}:A\to C$ be a scenario (or protocol), let ${\cal B}$ be a slice included in it with $B:={\cal B}$, and let $\psi: I\to A$ be a state.   The  \em local state at ${\cal B}$ relative to ${\cal G}$\em, where $\sigma\circ{\cal G}: I\to B\otimes B'$  with $g:={\cal G}$ is a sub-scenario   of ${\cal F}\circ \psi$, is the state $(1_B\otimes \top_{B'})\circ\sigma\circ g$.  
\end{defn}

With ${\cal F}$ and ${\cal B}$ as in picture (\ref{pic:sliceII}), for subscenario ${\cal G}$ the annotated part of ${\cal F}\circ \psi$\,: 
\[
\raisebox{-12mm}{\epsfig{figure=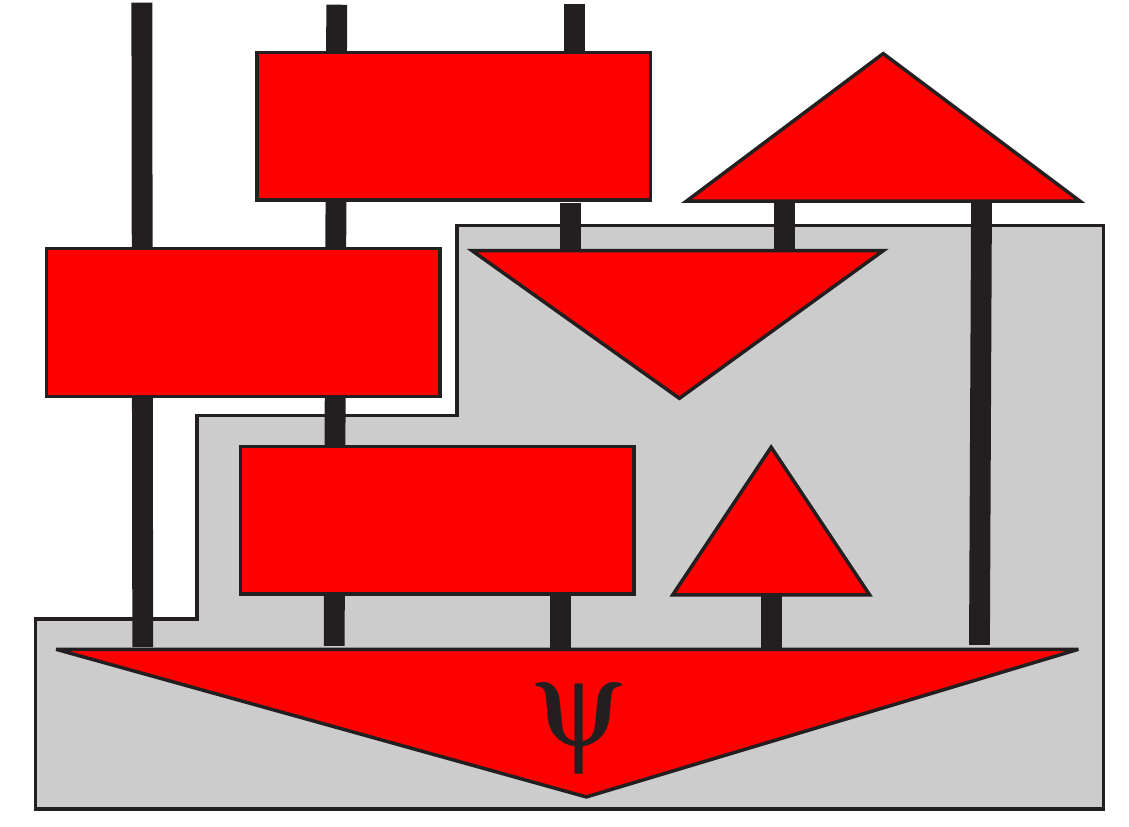,width=115pt}} \ ,
\]
the local state is:
\[
 (1_B\otimes \top_{B'})\circ g = \raisebox{-12mm}{\epsfig{figure=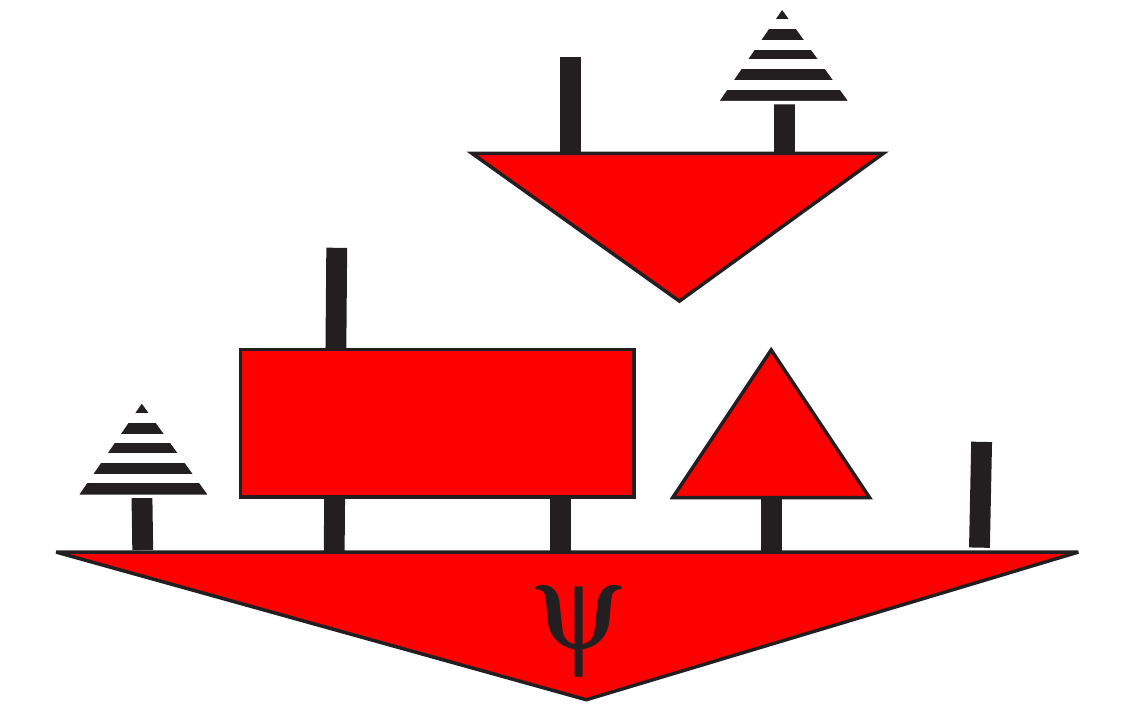,width=115pt}} 
\]
where we relied on eq.~(\ref{eq:symnat2}) to eliminate the symmetry isomorphisms. 

\begin{thm}[Existence and uniqueness of local states]\label{thm:LocalState}
For an SMC with terminal unit object, a slice admits a local state for any scenario in which it is included if and only if it is a spatial slice.  
Moreover, local states, whenever they exist,  do not depend on the choice of the subscenario ${\cal G}$ of Definition \ref{def:localstate}. \em
\end{thm}
\begin{proof}
First we show that if a slice is non-spatial, then there exists a scenario ${\cal F}$ for which ${\cal B}$ does not admit any local states.  If ${\cal B}$ is non-spatial then there exists $B_i$ and $B_{j\not = i}$ such that some morphism $f:B_i\to B_j$ is connected.  Setting:
\[
{\cal F}=  f\otimes 1_{\otimes_{k\not=i, j} B_k}=\ \raisebox{-5mm}{\epsfig{figure=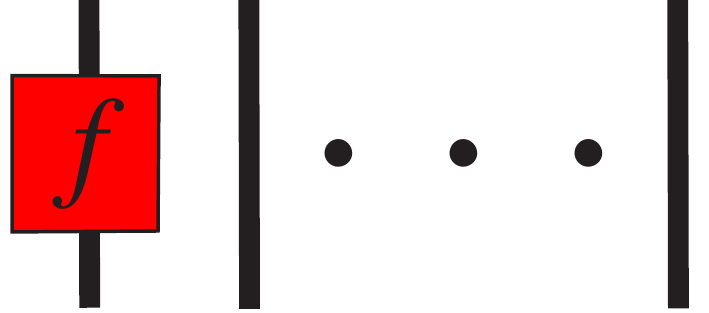,width=75pt}}
\]
one easily sees that  ${\cal F}\circ \psi$ admits no subformula ${\cal G}$ of the type required to yield a local state. 

If ${\cal B}$ is spatial, then for any scenario ${\cal F}$  which includes ${\cal B}$ we can construct the \em causal past \em of ${\cal B}$ as follows. First, consider all morphisms in  ${\cal F}$ for which an object in the outcome type is in ${\cal B}$; if such an object happens to be part of the input type of ${\cal F}$ then we can consider the identity morphism.  Then denote by ${\cal D}$ the slice consisting of all the input types of these morphisms, and now repeat this procedure, with ${\cal D}$ playing the role of ${\cal B}$, until we obtain a slice ${\cal Z}$ consisting of objects in the input type of ${\cal F}$.  The causal past then consists of  all the morphisms that we encountered in this procedure, together with $\sigma^{-1}\circ(1_Z\otimes \top_{Z'})\circ\sigma\circ\psi$
for $Z:={\cal Z}$ and $Z\otimes Z'$ up to symmetry equal to the input type of ${\cal F}$; e.g.:
\[
\raisebox{-12mm}{\epsfig{figure=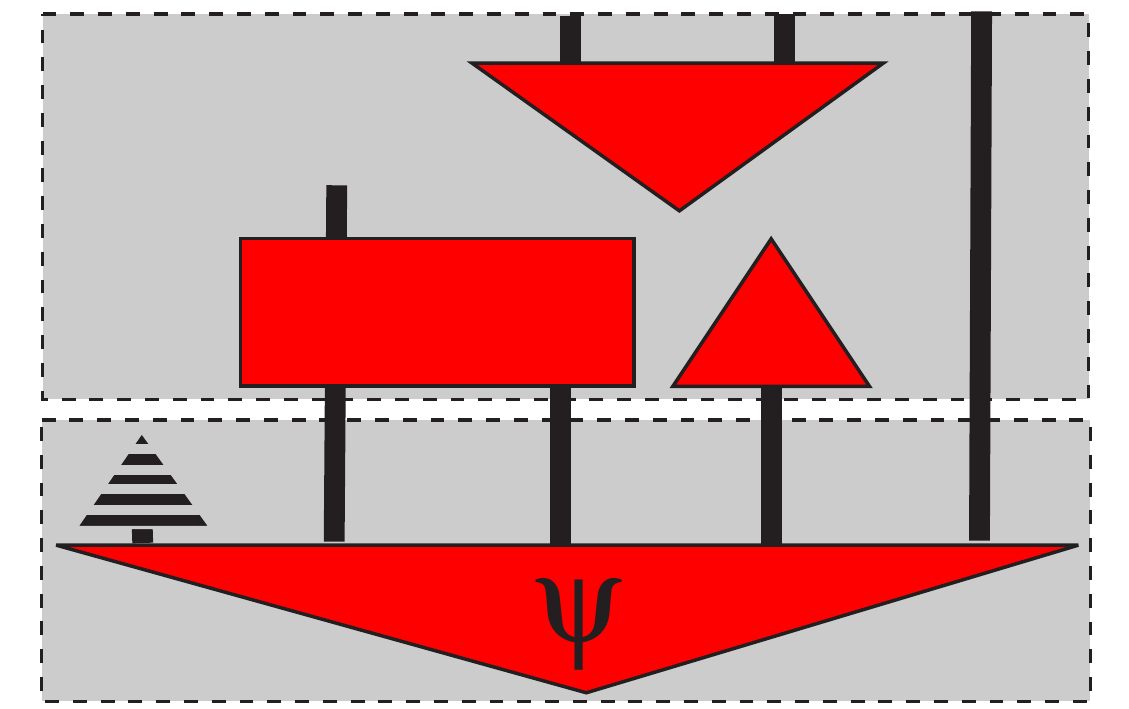,width=115pt}}
\]
in the case of the example in picture (\ref{pic:sliceII}). One then obtains the local state by post-composing  any object $E$ in the output type which is not in ${\cal B}$ with $\top_E$. 

We now show independence of the local state on the choice of the subscenario ${\cal G}$.  Any such ${\cal G}$ will include all the morphisms accounted for in the causal past, precisely by the construction of the causal past. We now proceed by induction on the number of morphisms contained in ${\cal G}$. We can enlarge ${\cal G}$ in two manners, either by means of sequential or parallel composition with some morphism $h:D\to E$, respectively yielding $h\otimes  {\cal G}$ or $(h\otimes 1_{D'})\circ{\cal G}$, where we omitted symmetry isomorphisms.  But since post-composition with $\top_E\otimes 1_{B\otimes B'}$ and $\top_E\otimes 1_{D'}$ respectively yields $(\top_E\otimes 1_{B\otimes B'})\circ (h\otimes  {\cal G})=\top_E\otimes {\cal G}$ and $(\top_E\otimes 1_{D'})\circ ((h\otimes 1_{D'})\circ{\cal G})=(\top_E\otimes 1_{D'})\circ {\cal G}$ the resulting local state will be the same as for ${\cal G}$.
\end{proof}

Theorem \ref{thm:LocalState} can be understood as arising from the way in which the object and morphism languages interact. Roughly speaking, the
morphism-language  defines how processes can be \em composed\em; the
interaction of the morphism-language with the object-language defines
how processes or scenarios can be \em decomposed \em using a slice; and for this latter structure to allow local states to be defined for each slice, we require a partial monoidal structure.

\begin{remark}
If we were to allow $f:B_i\rightarrow B_j$ to be a disconnected morphism in Theorem \ref{thm:LocalState} then the proof would break down, i.e. we can indeed define a local state. This is possible because, as per our assumption that morphism formulae are equivalent if they correspond to equivalent diagrams in the graphical language, a disconnected morphism $p\circ\top_{B_j}$ can be written as $\top_{B_j}\otimes p$, since this is in the same diagrammatic equivalence class. But this would then provide a subformula $\mathcal{G}:I\rightarrow B_i\otimes B_j\otimes B'$, which ensures a local state can be given at the spatial slice $B_i\otimes B_j$. In constrast, in the connected case we were not able to `push' $B_1$ next to $B_2$ to form $B_i \otimes B_j$---as can be done for the disconnected case---to be part of the codomain of $\mathcal{G}$.
\end{remark}

\begin{remark}
Note that from the proof of Theorem \ref{thm:LocalState} it also follows that in Definition \ref{def:localstate} we do not always need to specify  an initial state for the entire slice ${\cal A}$, but only for the slice ${\cal Z}$ which is included in the causal past of ${\cal B}$.
\end{remark}


\begin{defn}\em
A spatial slice ${\cal B}$ with $B:={\cal B}$ is \em total \em for a scenario ${\cal F}:A\to C$ in which it is included  
if ${\cal F}$ decomposes into two subscenarios ${\cal F}_1:A\to B$ and ${\cal F}_2:B\to C$.
\end{defn}

Total slices allow one to model evolution of a state through a scenario, when considering local states for `propagating' family total slices e.g.:
\[
\raisebox{-9mm}{\epsfig{figure=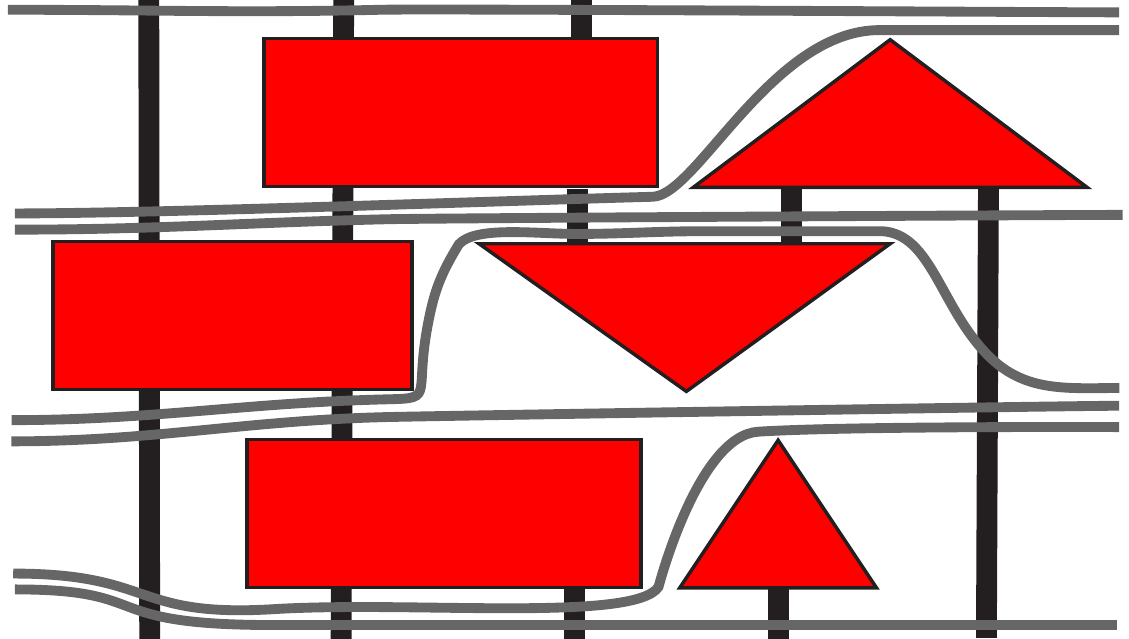,width=115pt}} 
\]
In this context, a \em general covariance theorem \em is one which states that the state of a system does not depend on the  particular choice of \em foliation\em, i.e.~the slice it belongs to.

\begin{cor}[general covariance]
Local states do not depend on the choice of foliation.
\end{cor}

We provide a simple example:  for the scenario and total spatial slices:
 \[
\raisebox{0mm}{\epsfig{figure=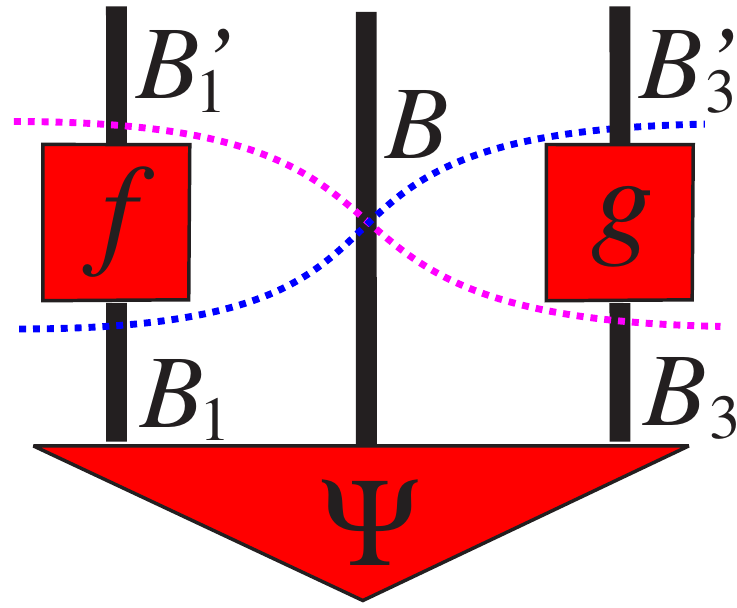,width=75pt}}
\]
we have:
 \[
\raisebox{0mm}{\epsfig{figure=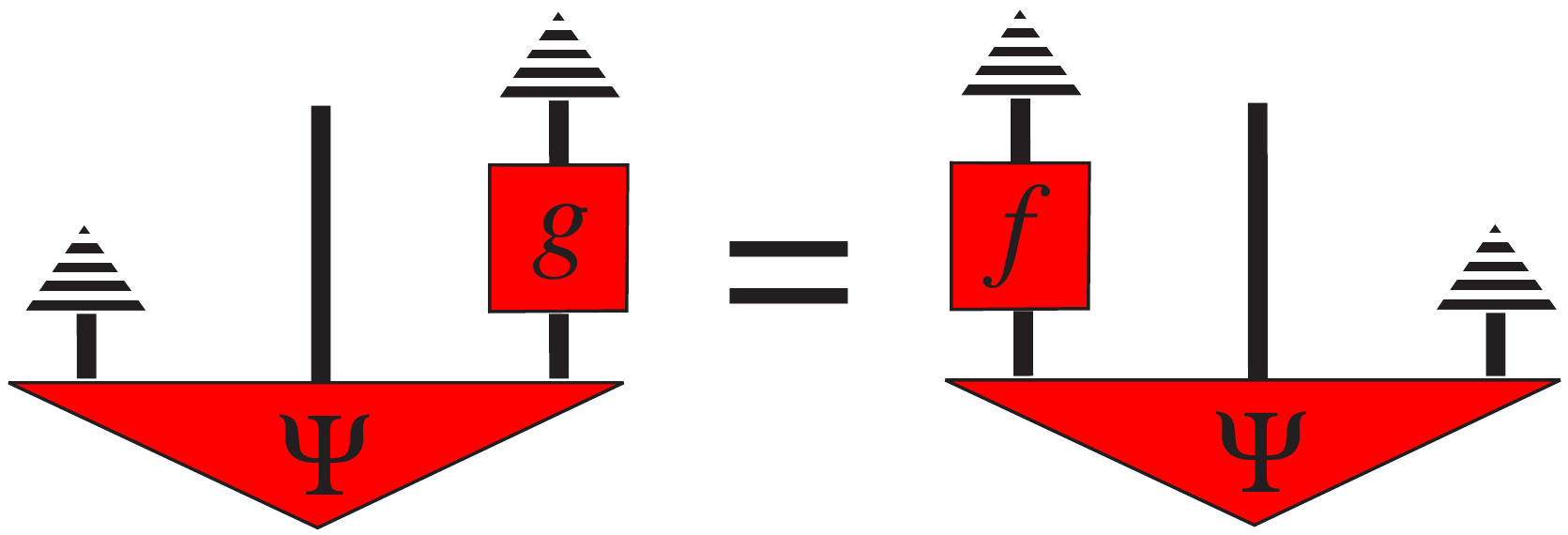,width=170pt}}
\]
since by terminality:
 \[
\raisebox{0mm}{\epsfig{figure=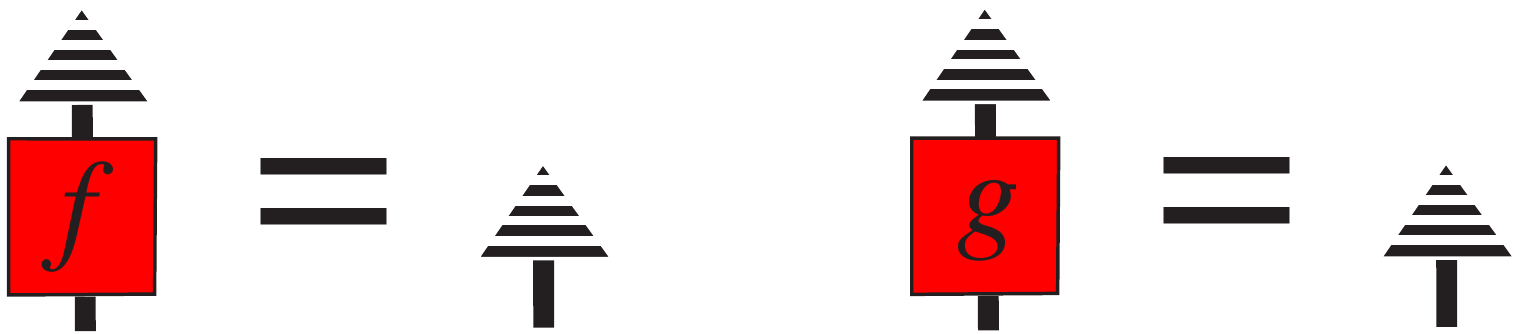,width=155pt}}
\]

\paragraph{Earlier work: BIP; Markopoulou; and Hardy.}   There are three strands of work which are directly relevant to
this Section.
\ben
 
\item Markopoulou \cite{Markopoulou}  defined a \em quantum causal history (QCH) \em as a mapping that assigns a Hilbert space $\mathcal{H}(x)$ to 
each element $x$ of a causal set $\mathcal{C}$ \cite{Sorkin}; it assigns tensor products $\mathcal{H}(\xi):=\mathcal{H}(x)\otimes\mathcal{H}(y)$ for an antichain $\xi=\{x,z\}$; and it assigns unitary mappings between antichains. This structure is clearly very similar to the spatial slices we have defined in monoidal categories. However, the assignment of unitaries between antichains $\xi$ and $\varsigma$ does not take into account the causal structure between elements $x\in\xi$ and $w\in\varsigma$. Hence a unitary $U:\mathcal{H}(\xi)\rightarrow \mathcal{H}(\varsigma)$ for which the state at $w$ is not a constant function of the state at $x$ is allowed even if $x$ and $w$ are not causally related. Therefore a QCH cannot enforce the fact that teleportation without classical communication cannot provide information flow.

\item Blute, Ivanov and Panangaden (BIP) developed a mathematical framework for describing the evolution of open quantum systems in  \cite{BIP}, that is related to Markopoulou's work and the causal sets programme. A similar notion of spatial slice exists for that framework, and it is shown there that the states are covariant. However, whereas in a causal category covariance follows immediately from the definition of the category, demonstrating covariance in \cite{BIP}  is much more involved, relying on properties of the concrete model of density matrices on Hilbert spaces.
\item 
Hardy has developed an operational framework for describing general probabilistic theories \cite{HardyPicturalism}. His framework also uses diagrammatic methods, although its formalization is not explicit. His work uses the way in which inputs and outputs of boxes are connected to define the causal structure of a scenario. This is similar to our work, because we define causal structure (see Subsection \ref{sec:CausalityInSMC}) using the \em processes \em available between two slices, i.e.~whether or not they are disconnected. However in a causal category the causal structure can be defined globally, whereas in Hardy's approach the causal structure is defined only for the boxes for which connections are made.
\een
\medskip


 Above we showed that terminality of the tensor unit implies covariance; however, in a CJ-universe the converse is also true, as we now show.

\begin{thm}
In a CJ$\top$-universe, general covariance implies terminality of the tensor unit.
\end{thm}
\begin{proof}
When considering
 \[
\raisebox{-8mm}{\epsfig{figure=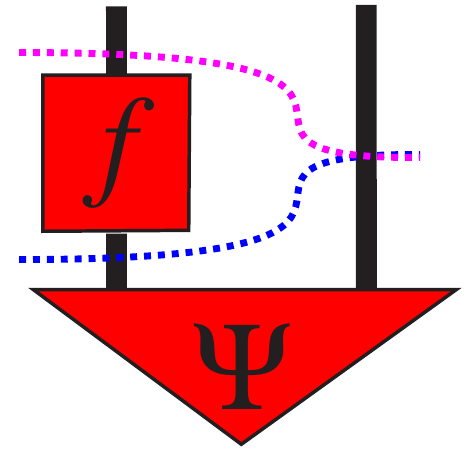,width=46pt}}
\qquad\mbox{covariance means}\qquad
\raisebox{-8mm}{\epsfig{figure=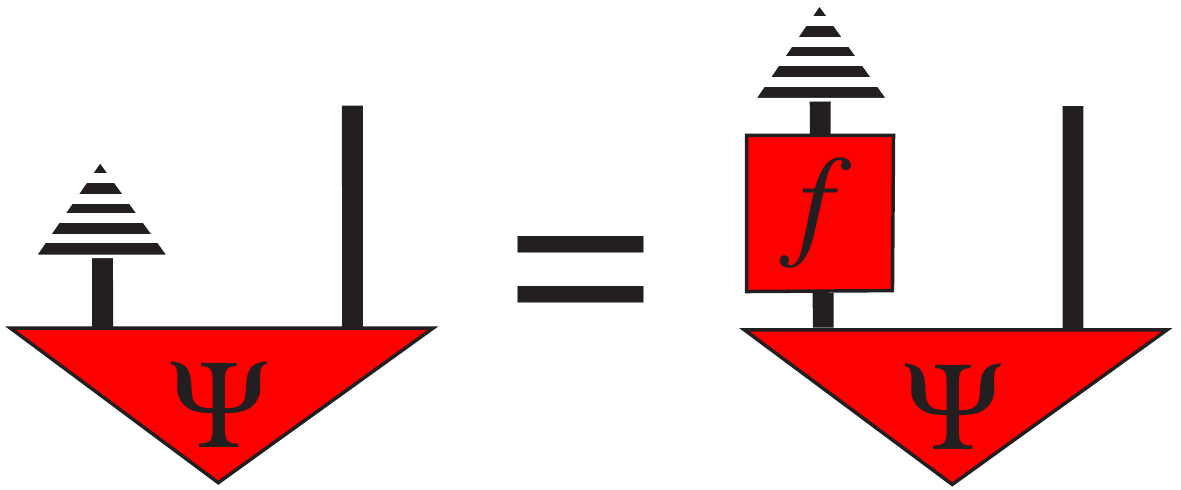,width=120pt}}
\]
which by injectivity in the definition of CJ-universe requires $\top_A=\top_B\circ f$.
\end{proof}

Hence, slices which are not spatial will not allow us to meaningfully describe the local state on some part of the slice. Therefore, to ensure that this is always possible:
\bit
\item \em we will restrict the tensor to causally unrelated (i.e.~disconnected) systems\em.  
\eit
Our rigorous formal definition is in the next Section.  One key consequence of this is that: 
\bit
\item \em all systems (i.e.~objects) in a causal category correspond to spatial slices,\em 
\eit
without any further do, simply by the compositional structure.  

\begin{remark}[Crossing slices]\label{rmk:crossing}
Although we shall restrict tensor composition of objects  we  will not restrict tensor composition of processes.  In contrast to other work in the same vein \cite{Markopoulou, BIP}, this will allow for processes to be defined between `crossing' slices.  For example, for slices $A\otimes B$ and $C\otimes D$, with $A$ causally preceding $C$ while $D$ causaly preceeds $B$, it still makes perfect sense to speak of processes of type $A\otimes B\to C\otimes D$, which will all be 
of the form 
\[
f\otimes (\top_C\circ \psi)=f\otimes \top_C\otimes \psi
\]
with $f$ arbitrary in ${\bf C}(A, B)$ and $\psi$ arbitrary in ${\bf C}({I}, D)$, graphically:
\[
\raisebox{0mm}{\epsfig{figure=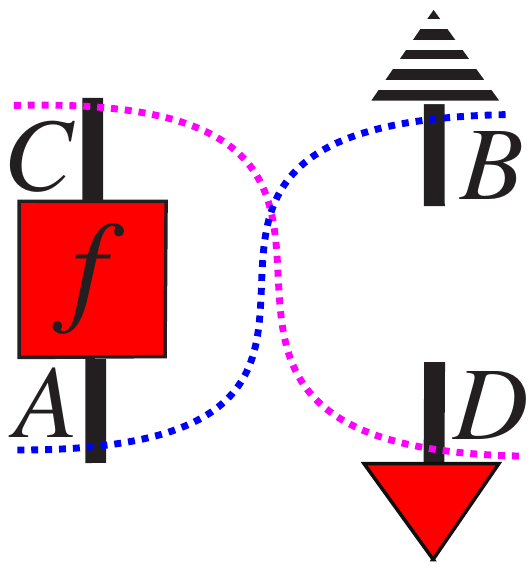,width=54pt}} 
\]
\end{remark}


%
%



\section{Definition and analysis of causal categories}\label{sec:Definition}

Given that, in a category ${\bf C}$, the absence of first-kind information-flow between objects $A$ and $B$ is implemented by making the hom-set $\mathbf{C}(A,B)$ disconnected, we can summarise the results of the previous Section in the table  below; it shows  the mathematical structure corresponding to the physical properties that we aim to axiomatize. 

\begin{center}
\begin{tabular}{|l|l|p{4cm}|}
Physical property & Mathematical structure & Assumptions \\
\hline

No second-kind info-flow & Terminality of tensor unit & Existence of CJ states \\
Unique local state for each slice  & Partial monoidal structure &  Terminality of tensor unit 
\end{tabular}
\end{center}

Note that our concrete models will typically be quantum theory or classical probability theory, so the assumption of the existence of CJ states will be satisfied. Since this assumption leads to teminality, 
we note that terminality is actually not an extra assumption given that we assume the existence of CJ states. This leads to the formal definition of a caucat which we now introduce.

\subsection{Definition of a causal category}

We use $[-]$ to denote pointwise application.  As already indicated above, we will  take natural isomorphisms of  monoidal structure to  be strict. 



\begin{defn}\em
A \em partial functor \em $F:\mathbf{B}\rightarrow \mathbf{C}$ is a functor $\hat{F}:\mathbf{A}\rightarrow \mathbf{C}$, where $\mathbf{A}$ is a subcategory of $\mathbf{B}$; $\mathbf{A}$ is called the \em domain of definition \em of $F$, written $\textrm{dd}(F)=\mathbf{A}$, and $\mathbf{B}$ is called the \em domain \em of $F$, written $\textrm{dom}(F)=\mathbf{B}$. A \em partial bifunctor \em is a partial functor whose domain is a product category.
\end{defn}


\begin{defn}\label{def:spmc}\em
A \emph{symmetric strict partial monoidal category} is a category $\mathbf{C}$, together with a partial bifunctor $\otimes:\mathbf{C}\times\mathbf{C}\rightarrow \mathbf{C}$, for which $\textrm{dd}(\otimes)$ is a full subcategory of $\textrm{dom}(\otimes)$, and such that there exists a \em unit object \em $I$, which is the unit of a partial monoid $(|\mathbf{C}|,\otimes,I)$:
\bit 
\item[(u1)] For all $ A\in |\mathbf{C}|$, both $(A,I) \in \textrm{dd}(\otimes)$ and $(I,A) \in \textrm{dd}(\otimes)$, and 
\item[(u2)]  $A\otimes I=A=I\otimes A$; 
\item[(a1)] $\forall A,B,C\in |\mathbf{C}|$, $(A,B),(A\otimes B,C)\in \textrm{dd} (\otimes)$ iff $(B,C),(A,B\otimes C)\in \textrm{dd} (\otimes)$,
\item[(a2)] when they exist, $A\otimes (B\otimes C)=(A\otimes B)\otimes C$, and 
\item[(a3)]  for any morphisms $f,g,h$ in $\mathbf{C}$, $(f\otimes g)\otimes h=f\otimes (g \otimes h)$ when they exist.

\item[(s1)] for all $ A,B\in |\mathbf{C}|, (A,B),(B,A)\in\textrm{dd}$, there exists a \em symmetry morphism \em
 
\[
\sigma_{A,B}:A\otimes B\rightarrow B\otimes A
\] 
such that $\sigma_{A,B}\circ\sigma_{B,A}=1_{A\otimes B}$.
\eit
\end{defn}

\begin{remark}[Associativity of parallel composition for morphisms]
 Since $\textrm{dd}(\otimes)$ is a full subcategory of $\textrm{dom}(\otimes)$, the partial monoidal product $f\otimes g$ of morphisms $f:A\rightarrow D$ and $g:B\rightarrow E$ exists iff $A\otimes B$ and $D\otimes E$ exist. Then, given a morphism $h:C\rightarrow F$, since $(A\otimes B)\otimes C$ exists iff $A\otimes (B \otimes C)$ exists, we also have that $(f\otimes g)\otimes h$ exists iff $f\otimes (g \otimes h)$ exists.
\end{remark}

\begin{remark}[Bifuncoriality]
For a partial monoidal category bifunctoriality holds just as for a (full) monoidal category, i.e.
\[
(h\otimes k)\circ(f\otimes g)=(h\circ f)\otimes(k\circ g).
\]
This is guaranteed by the fact that the domain of definition for a partial bifunctor is a full subcategory of its domain, so the monoidal product of the composites $h\circ f$ and $k\circ g$ always exists.
\end{remark}

\begin{example} Any strict monoidal category is a strict partial monoidal category, where $\textrm{dd}(\otimes)=\textrm{dom}(\otimes)$, and any category that contains a strict monoidal category as a full subcategory is a strict partial monoidal category.  \end{example}

\begin{remark} 
The symmetry morphism can be seen as a `kinematic' feature of a caucat, analogous to inversion of a spatial axis in a conventionally formulated physical theory.
\end{remark}



\begin{defn}\label{def:caucat}\em
A \emph{causal category} (or \em caucat\em) $\mathbf{CC}$ is a symmetric strict
partial  monoidal category whose unit object $I$ is terminal, i.e.~for each object $A\in|\mathbf{CC}|$ there is a unique morphism $\top_A:A\to I$,
and for which the monoidal product, $A\otimes B$, exists iff 
\beq\label{eq:defcaucat}
\mathbf{CC}(A,B)=[\mathbf{CC}(I,B)]\circ\top_{A}\quad\mbox{and}\quad\mathbf{CC}(B,A)=[\mathbf{CC}(I,A)]\circ\top_{B}\,.
\eeq
We also require that each object has at least one element, i.e.~$\forall A\in|{\bf CC}|: {\bf CC}(I, A)\not=\emptyset$.
\end{defn}

\begin{prop}
{\bf (i)} In a caucat, morphisms $f:A\to B$ are `normalized', i.e.~$\top_B\circ f = \top_A$.
{\bf (ii)} In a caucat $\top_I=1_I$ and $\top_{A\otimes B}=\top_A \otimes \top_B$ whenever $A\otimes B$ exists.
\end{prop}

\begin{proof}
Both {\bf (i)} and {\bf (ii)} follow straightforwardly by terminality of $I$.
\end{proof}

We shall refer to a monoidal or partial monoidal category in which the unit object is terminal, and for which each object $A$ has at least one element, as a \em normalized category \em (or \em normcat\em). The definition of disconnectedness for a partial monoidal category is the same as for a full monoidal category (cf.~Defn.~\ref{def:disconnected} and Defn.~\ref{def:disconnectedbis}), but we state it explicitly for the reader's convenience.

\begin{defn}[Disconnectedness for partial monoidal category]\em
In a partial monoidal category, a morphism  $f:A\to B$ is  \em disconnected \em if it decomposes as $f=p\circ e$ for some $e:A\to I$ and $p:I\to A$, and a  hom-set $\mathbf{C}(A,B)$ is  \em disconnected \em     if it contains only disconnected morphisms.  If both $\mathbf{C}(A,B)$ and $\mathbf{C}(B,A)$ are disconnected then we say that the objects $A$ and $B$ are disconnected.
\end{defn}

\begin{prop}\label{prop:properties}
{\bf(i)} In Definition \ref{def:caucat}, Eq.~(\ref{eq:defcaucat})  is equivalent to both $\mathbf{CC}(A,B)$ and $\mathbf{CC}(B,A)$  being disconnected. {\bf(ii)} In a caucat, $A\otimes B$ exists iff $B\otimes A$ exists. {\bf(iii)} Conditions (u1) and (a1) in the definition of partial monoidal category are implied by Eq.~(\ref{eq:defcaucat}) together with the condition that if $A\otimes B$, $A\otimes C$, $B\otimes C$ exist then also $A\otimes(B\otimes C)$ exists.

\end{prop}

 \begin{proof}
{\bf(i)} follows from the fact that by terminality of $I$ any disconnected morphism $f:A\to B$ in a caucat is of the form $p\circ \top_A$.  {\bf(ii)} follows straightforwardly from the symmetry of Eq.~(\ref{eq:defcaucat}).
Concerning {\bf(iii)}, since $I$ is terminal we have:
\bit
\item $\mathbf{CC}(I,A)=[\mathbf{CC}(I,A)]\circ 1_{I}=[\mathbf{CC}(I,A)]\circ\top_{I}$, and
\item $\mathbf{CC}(A,I)=\{\top_A\}=\{1_I\circ\top_A\}= [\mathbf{CC}(I,I)]\circ\top_{A}$,
\eit
so $I\otimes A$ and $A\otimes I$ always exist. If $(A\otimes B)\otimes C$ exists, then for $f\in{\bf CC}(A, C)$,  $p_B\in {\bf CC}(I, B)$:
\[ 
f= f\otimes (\top_B\circ p_B)=(f\otimes \top_B)\circ(1_A\otimes p_B)=(p_C\circ \top_{A\otimes B})\circ(1_A\otimes p_B)
=p_C\circ \top_A
\]
for some $p_C\in {\bf C}(I, C)$, so it follows that $A\otimes C$ exists, and by symmetry also $B\otimes C$ exists, and hence by our additional assumption $A\otimes (B\otimes C)$ also exists.
\end{proof}

\begin{remark}
Item (ii) in Proposition \ref{prop:properties} shows that the symmetry morphism can be consistently defined for a caucat.
\end{remark}

\begin{defn}\em
In a causal category:
\bit
\item If $A\otimes B$ exists then we call $A$ and $B$ \em space-like separated\em.
\item If $\mathbf{CC}(A,B)$ is connected while $\mathbf{CC}(B,A)$ is disconnected then $A$ \em  causally precedes \em $B$. 
\item  If $\mathbf{CC}(A,B)$ and $\mathbf{CC}(B,A)$ are connected then $A$ and $B$ are \em causally intertwined\em. 

\eit
\end{defn}

\begin{remark}
A pair of objects which are causally intertwined each correspond to a spatial slice, but they are the  `crossed' spatial slices that we discussed in Remark  \ref{rmk:crossing}.
\end{remark}

\begin{example}
Each category induces a caucat by freely adjoining a monoidal unit (and a state for each object); we could call such a degenerate caucat \em purely temporal\em.  Each monoid induces also another caucat with the monoid as the tensor by freely adoining a unique morphism for each ordered pair of objects; we could call such a degenerate caucat \em purely spatial\em. 
\end{example}

\begin{remark}
Since in non-degenerate situations identities are  connected, the tensor of $A$ with itself will typically not exist (see  Lemma~\ref{thm:factors}  below).
\end{remark}


\subsection{Relation to dagger compact categories}\label{sec:Incompatibilities}
 
We now show that some basic aspects of CQM, involving identical or isomorphic objects (which will allow us to identify systems of the same kind), in particular the use of compactness and dagger structure, are incompatible with the caucat structure! In Subsection \ref{sec:recover} we shall describe how these aspects can be reinstated.

We first show that isomorphisms cannot be used to represent the property that two systems at different spatiotemporal locations are of the same type (e.g. a qubit).

\begin{prop}\label{thm:isos}
Given a caucat $\mathbf{CC}$, suppose that $A$ causally precedes $B$, or that $A$ and $B$ are causally unrelated. Then if $A\cong B$, it follows that $A\cong I\cong B$. 
\end{prop}

\begin{proof}
If either $A$ causally precedes $B$, or $A$ and $B$ are space-like separated, then  $\mathbf{CC}(B,A)$ is disconnected. Hence for the iso $f:A\rightarrow B$ we have, for some $\psi:I\to A$,
\[
1_A = f^{-1}\circ f = \psi\circ\top_B\circ f = \psi\circ \top_A
\]
Since by terminality of $I$ we have $\top_A\circ \psi=1_I$, we obtain $A\cong I$, and $B\cong I$ similarly.
\end{proof}

Hence the fact that systems at different spacetime locations are of identical types cannot be witnessed in the caucat, but instead in the $\dagger$-SMC that will be used to construct the caucat---we describe how this can be done in Subsection \ref{sec:recover}.

As discussed in Section \ref{sec:partial}, the tensor product of a system with itself is not meaningful in a causal setting. This can be formalised as follows. 

\begin{lem}\label{thm:factors}
If the identity morphism $1_A$ on an object $A$ in a caucat $\mathbf{CC}$ is disconnected, then $A$ has only one state. Any morphism between objects for which the identity is disconnected  on either the domain or codomain  is also disconnected.
\end{lem}
\begin{proof}
If  $1_{A}=\psi\circ\top_{A}$ for some $\psi:I\rightarrow A$,
then for any other state $\phi:I\rightarrow A$ we have 
\[
\phi=1_{A}\circ\phi=\psi\circ\top_{A}\circ\phi=\psi\circ1_{I}=\psi.
\]
Consider a morphism $f:A\rightarrow B$ for which $1_A=\psi\circ\top_A$. Then 
\[
f=f\circ 1_A=f\circ \psi \circ \top_A=\phi\circ \top_A
\]
for some morphism $\phi:I\rightarrow B$.  The codomain case proceeds similarly.
\end{proof}

\begin{prop}
 If $A\otimes A$  exists then $1_A$ is disconnected.
\end{prop}

\begin{proof}
If $A\otimes A$  exists then $\mathbf{CC}(A,A)=[\mathbf{CC}(I,A)]\circ\top_{A}$, and $1_A$ is disconnected.
\end{proof}

We restate Corollary \ref{cor:disc} for completeness, since it shows that there are no non-trivial objects with compact structures.

\begin{prop}\label{thm:collapse}
For an object  $A$ in a caucat  with a compact structure $1_A$ is disconnected, and morphisms between objects with a compact structure are disconnected. Hence for a compact subcategory of a caucat all morphisms are disconnected.
\end{prop}



Hence compact structure is incompatible with the structure of a caucat. Dagger structure can also not be maintained.

\begin{prop}\label{thm:dagger}
In a caucat with a dagger functor every object has only one state, and hence compound objects $A\otimes B$ only have disconnected states.  
\end{prop}

\begin{proof}
For a given object $A$, a dagger functor provides a bijection 
\[
 \mathbf{CC}(I, A)\cong \mathbf{CC}(A, I)
\]
so  since $I$ is terminal this can only occur if each object has only one state. 
\end{proof}

\section{Constructing caucats}\label{sec:constructCC}

In this Section we shall describe methods for constructing caucats. The first step consists of normalizing a ($\dagger$-compact)  SMC. Then there are two options, either to   `carve out' an appropriate subcategory  of a monoidal category---this will represent discarding the unphysical objects in the category (when a causal structure is already given in the category)---or to combine it with a causal structure, resulting in a partial monoidal product that exists for pairs of objects which are not causally related. 

\[
\xymatrix{ & \dagger\textrm{-compact category}\ar[d]^{\textrm{normalize}} & \\
 & \textrm{Normalized category}\ar[dl]^{ \textrm{carving}}\ar[dr]_{\textrm{pairing}}  & \\
 \textrm{Caucat}&  &  \textrm{Caucat}}
\]    

Finally in Subsection \ref{sec:recover} we shall describe how to reinstate the power of CQM, given that we showed above that structures such as compactness cannot be retained in caucats. 

\subsection{Normalizing}\label{sec:normalising}

\begin{defn}[Normalization]\em
Given a ($\dagger$-compact)  SMC $\mathbf{C}$ with environment structure, we define a new SMC $\mathbf{C}_\top$ with the same objects as $\mathbf{C}$ but with morphisms restricted to normalized ones, and a corresponding inclusion functor $F_\top:\mathbf{C}_\top\hookrightarrow \mathbf{C}$.
\end{defn}

Evidently, if $\mathbf{C}$ is $\dagger$-compact, $\mathbf{C}_\top$ will not be, therefore we retain its connection with the given $\dagger$-compact  SMC via the strict monoidal functor  $F_\top$.  But while  we cannot retain dagger structure and compact structure, we can retain a \em conjugate functor\em, which can be constructed from the  dagger structure and the compact structure.  

Firstly, given a compact category $\mathbf{C}$ (see Definition \ref{def:compact}) we can define a contravariant functor:
\[
(-)^*: \mathbf{C}\to \mathbf{C}:: f\mapsto (1_{A^*}\otimes \eta_B)\circ (1_{A^*}\otimes f \otimes 1_{B^*})\circ(\eta_A\otimes 1_{B^*})\,,
\] 
that is, in diagrams, \vspace{-1.5mm}
\beq\label{pic:transpose}
\raisebox{-7mm}{\epsfig{figure=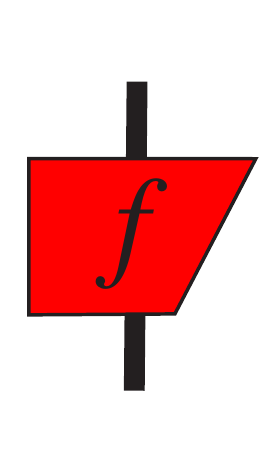,width=26pt}}\ \mapsto\
\raisebox{-7mm}{\epsfig{figure=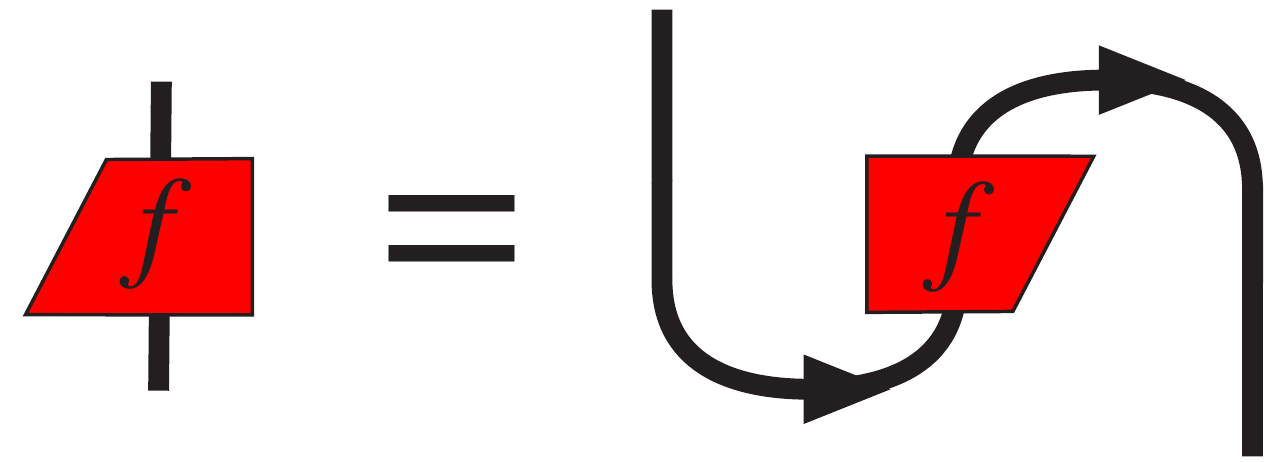,width=120pt}}\vspace{-1.5mm} \,,
\eeq
where we used a 180-degree rotation of the box representing the morphism to denote its \em transpose\em.  If we moreover have a dagger functor we can also define a covariant functor \cite{AC1}:
\[
(-)_*= ((-)^*)^\dagger=((-)^\dagger)^*: \mathbf{C}\to \mathbf{C}:: f\mapsto (1_{B^*}\otimes \eta_A)\circ (1_{B^*}\otimes f^\dagger \otimes 1_{A^*})\circ(\eta_B\otimes 1_{A^*})\,,
\] 
that is, in diagrams, \vspace{-1.5mm}
\beq\label{pic:conjugate}
\raisebox{-7mm}{\epsfig{figure=Transpose0.pdf,width=26pt}}\ \mapsto\
\raisebox{-7mm}{\epsfig{figure=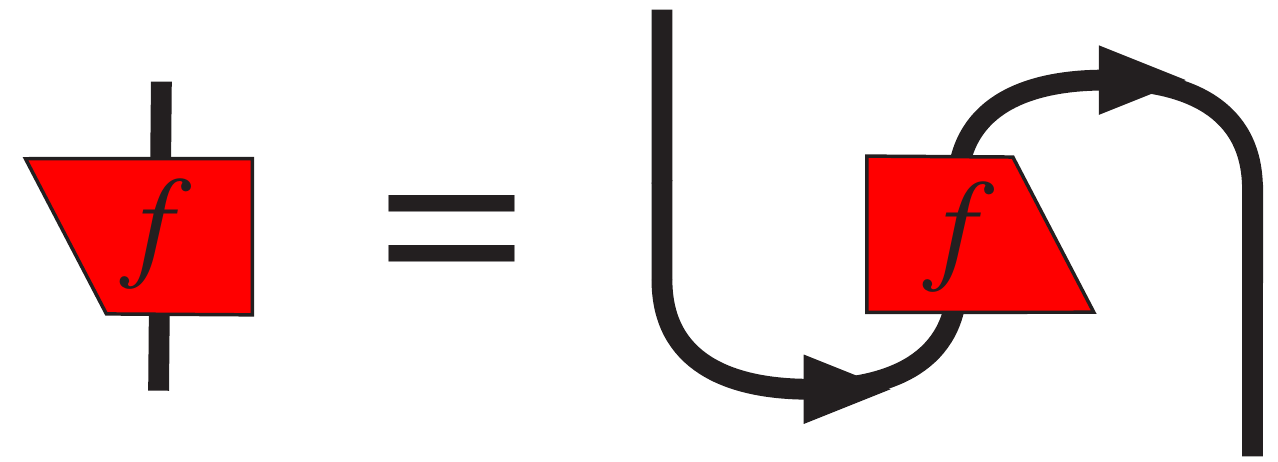,width=120pt}}\vspace{-1.5mm} \,,
\eeq
where we use reflection in the vertical axis to denote the \em conjugate\em. 

Under mild assumptions, also  the existence of a CJ-state is also retained in a caucat.


\begin{thm}\label{ref:thmDagCaucat}
Let $\mathbf{C}$ be a dagger compact category with an environment structure for which, for all $A\in|\mathbf{C}|$:
\bit
\item[{\bf (cc1)}] the scalar $s_A=\top_{A^*\otimes A}\circ\eta_{A}: I\to I$ is invertible, and,
\item[{\bf (cc2)}] $(\top_A)_*=\top_{A^*}$\,,
\eit
 then  in  $\mathbf{C}_\top$ every object has a  CJ-state $\eta_A^{CJ}=(s_A)^{-1}\cdot\eta_{A}$,  and the  conjugation functor from $\mathbf{C}$ is inherited. 
\[
\qquad 
\] 
\end{thm}
\begin{proof}
For $A\in | \mathbf{C}_\top|$ the morphism $(s_A)^{-1}\cdot\eta_{A}\in\mathbf{C}(I, A\otimes A)$ is normalized by assumption (cc1), so in $ \mathbf{C}_\top$, and the fact that this is a CJ-state follows straightforwardly from compactness, since the property of being a CJ-state is weakened form of compactness.  Moreover, if $f:A\to B$ is normalized, then
 \[
\begin{array}{rcll}
\top_{B_*}\circ f_* & = & (\top_B)_*\circ f_* & \mathrm{(cc2)}\\
& = & (\top_B\circ f)_* & \mathrm{(functoriality)} \\
& = & (\top_A)_* & \mathrm{(normalization)} \\
& = & \top_{A^*} & \mathrm{(cc2)}.
\end{array}
\]  
So conjugates of normalized morphisms are also normalized, and hence $ \mathbf{C}_\top$ inherits the conjugation
functor from $\mathbf{C}$. 
\end{proof} 

\subsection{Causal structure and carving}\label{sec:CausalStructure}

Causal structure has been studied intensively in relativity, and prominent examples of related work include the  \em Kronheimer-Penrose axioms for space-time \em \cite{KP} and Sorkin's \em causal set \em (causet) program, surveyed in \cite{Sorkin}. The basis of these axiomatizations are Zeeman's theorem \cite{Zeeman}, which states that the Poincar\'{e} group can be defined as the group of automorphisms of Minkowski space that preserve the causal order on points (i.e.~space-time events), and  Malament's theorem \cite{Malament}, which states that  for temporally oriented spacetimes which are past- and future-distinguishing  both the differentiable structure and the conformal metric can be recovered from the causal structure.   Other more recent work includes the identification of a domain-theoretic structure on space-time intervals for an important class of space-times  by Martin and Panangaden \cite{KeyePrakash}. 

We define causal structure in a  monoidal category as follows.

\begin{defn}\label{def:causal}\em
The \emph{causal structure} of a partial (or full) monoidal category is a directed graph $\mathcal{G}=(G,E)$, whose vertices $G$ are the objects of $\mathbf{CC}$, and an edge $(A,B)\in E$ exists iff $\mathbf{CC}(A,B)$ is connected.
\end{defn}



\begin{example}\label{ex:causal}
We can define a caucat $\mathbf{CC}$, shown in the diagram below, whose causal structure is the directed graph of a `3-loop'. The caucat can be obtained from the graph by freely adjoining the monoidal unit.
\[
\xymatrix{ & I\\
A\ar@{-->}[ur]\ar[rr] &  & B\ar[dl]\ar@{-->}[ul]\\
 & C\ar@{-->}[uu]\ar[ul]}
\]
In this caucat, the only pairs of objects for which $\otimes$ exists are $A\otimes I$, $I\otimes A$ and $I\otimes I$. The restrictions on the morphisms are as follows. Firstly, for related pairs $(A,B)$, we have $\mathbf{CC}(A,B)\neq[\mathbf{CC}(I,B)]\circ\top_{A}$. Secondly, we must ensure that any pair of composable connected morphisms $f,g$, the composite $g\circ f$ is disconnected. This is allowed, since there is nothing in the definition of a caucat that forces the composition of connected morphisms to be a connected morphisms. That is, connectedness of hom-sets is not transitive.
\end{example}

Since every (full) monoidal category is also a partial monoidal category, we can apply Definition \ref{def:causal} as follows. We denote a normalized subcategory of $\mathbf{C}$ by $\mathbf{C}_\top$.

\begin{example}[Degenerate causal structure]\label{ex:cpm}
Recall the category of mixed states and CP maps $\mathbf{Mix}$, defined in Example \ref{ex:mixedQT}: this is the setting for mixed-state quantum theory. We shall consider the subcategory of $\mathbf{Mix}$  containing only normalized  morphisms, $\mathbf{Mix}_\top$, i.e. its morphisms are now \em trace-preserving \em complete-positive maps. As described above, requiring each morphism to preserve the trace means that this category has a terminal unit object. Hence $\mathbf{Mix}_\top$ is a normalized category, and therefore satisfies some conditions of being a caucat. However, this category has \em degenerate \em causal structure: each object is connected to one another. For example, in each hom-set $\mathbf{Mix}(\mathbb{C}^{\otimes m},\mathbb{C}^{\otimes n})$ such that $m\neq n$ and $m,n>1$, there is a morphism that does not factor through $\mathbb{C}$. Indeed, this category actually represents the `opposite' structure of a causal category: every pair of objects in $\mathbf{Mix}_\top$ is connected but the tensor also exists for every pair of objects.


\end{example}



The preceding Example leads to the idea that to obtain a caucat from a normalized category, such as $\mathbf{Mix}_\top$, we must discard the unphysical objects. We can do this by defining a caucat as a subcategory of an SMC.

\begin{defn}[Carved category]\label{def:carving}\em
Let $\mathbf{C}$ be an SMC. A \em carved category (of $\mathbf{C}$)\em, denoted  $\mathbf{CC}$, is a normalized subcategory of $\mathbf{C}$ such that for $A,B\in|\mathbf{C}|$, iff $A$ and $B$ are connected then $A\otimes B\notin|\mathbf{CC}|$.
\end{defn}


To construct a caucat from a protocol and a category of physical processes, we proceed as follows. Recall that a directed graph $\mathcal{G}=(G,E)$ is \em transitive \em iff the adjancency relation $E$ is transitive. We want to extract a physical universe $\mathbf{CC}$ from $\mathbf{C}$, e.g.~some systems such as qubits, and define causal relationships between them using $\mathcal{G}$.

\paragraph{Carving construction.}
Let $\mathbf{C}$ be an SMC of possible physical processes, let $\mathbf{C}_\top$ be its normalized subcategory, and let $\mathcal{G}$ be a transitive directed graph. The \em carving construction of $\mathbf{CC}$ \em is, a triple $(\mathbf{C}, \mathcal{G},\kappa:G\rightarrow |\mathbf{C}|)$, where $\kappa$ is an injective mapping on the vertices of $\mathcal{G}$, which yields $\mathbf{CC}$ as a carved category of $\mathbf{C}$ via the following steps. 

\begin{enumerate} 
\item The mapping $\kappa$ identifies the objects of $\mathbf{C}$ which we want to consider in our protocol or universe. These objects will be the objects of our caucat $\mathbf{CC}$, i.e.~$\kappa(G) = |\mathbf{C}|$.  
\item We define the hom-sets of $\mathbf{CC}$  to be connected, and thus identical to the hom-sets of $\mathbf{C}_\top$, only when they are causally related in $\mathcal{G}$ (i.e. only when they are causally related in the intended protocol). Hence for $A=\kappa(g_1)$ and $B=\kappa(g_2)$, we have 
\[
\mathbf{CC}(A,B) := \left\{ 
\begin{array}{ll}
      				\mathbf{C}_\top(A,B) 		& \textrm{if } (g_1,g_2)\in E;\\
				 \left[\mathbf{C}_\top(I,B)\right]\circ\top_A       & \textrm{otherwise}.
				\end{array} 
\right.
\]
Transitivity of $\mathcal{G}$ ensures that the composition law of $\mathbf{CC}$ can be defined as the composition law of $\mathbf{C}$.
\item Finally, we make $\mathbf{CC}$ into a carved category of $\mathbf{C}$ by defining $\mathbf{CC}$  to have a partial monoidal product such that the monoidal product of a pair of objects $(\kappa(g_1),\kappa(g_2))$ exists only when  
$(g_1,g_2)\notin E$.
\end{enumerate}

\begin{remark}
Note that if the digraph $\mathcal{G}$ is not transitive it is always possible to define an $E'\supset E$ such that $\mathcal{G}$ induces a transitive digraph $\mathcal{G'}=(G,E')$.
\end{remark}

An advantage of carving is that there exists a canonical embedding into a $\dagger$-SMC, i.e.~the category out of which the caucat was carved. We discuss this in Subsection \ref{sec:recover}

\begin{example}[Teleportation protocol]
We can construct a causal category $\mathbf{T}$ for the teleportation (with classical communication) protocol using the carving construction. We define a directed graph $\mathcal{T}=(T,E)$ that represents the causal relationships in the protocol:
\[
\xymatrix{  &  a_5 \\
{a_3}\ar[ur] & {a_4}\ar[u] \\
{a_1}\ar[u] &  {a_2}\ar[u]\ar[ul]
}
\] 
This should be compared to the depiction in the graphical calculus:
\[
\raisebox{0mm}{\epsfig{figure=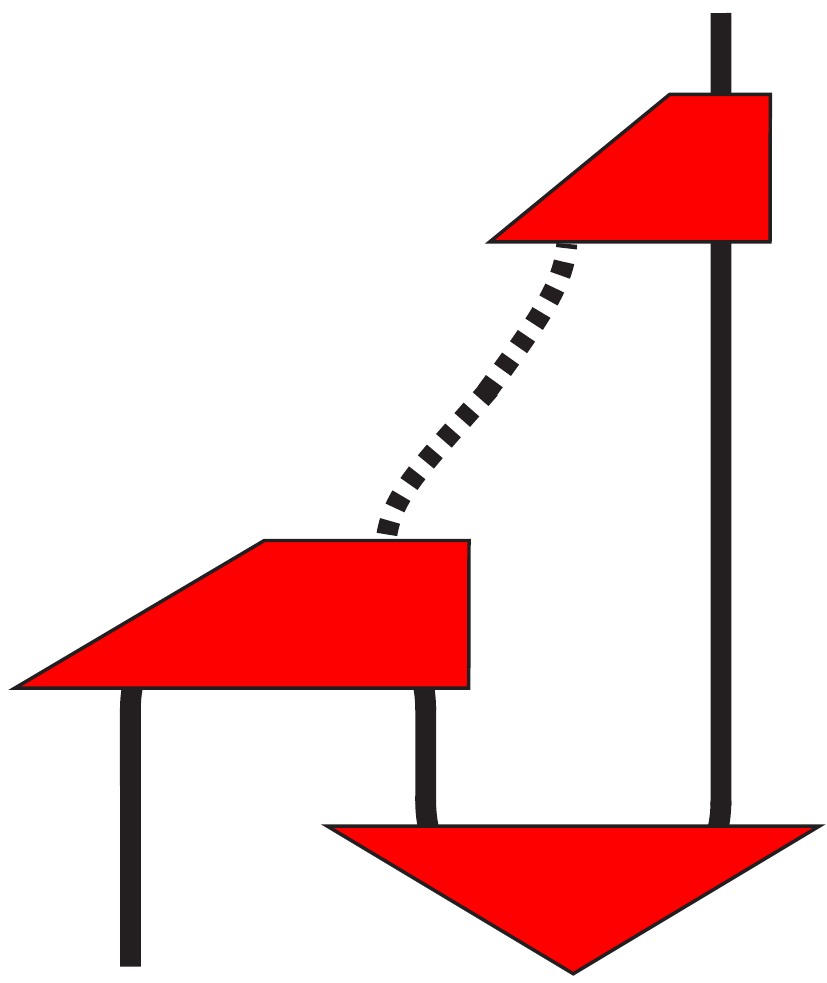,width=80pt}}
\]
We shall carve  $\mathbf{T}$ from $ \mathbf{Mix}$, using the carving construction $(\mathbf{Mix}, \mathcal{T},\kappa:T\rightarrow |\mathbf{C}|)$ as follows:
\begin{enumerate}
\item We take the requisite number of isomorphic copies of $\mathbb{C}^2$, or tensor products thereof, to be the objects of the category $\mathbf{T}$, so that the mapping $\kappa$ yields $\kappa(a_i)\cong \mathbb{C}^2$ for $i\in\{1,4,5\}$, and $\kappa(a_i)\cong \mathbb{C}^2\otimes \mathbb{C}^2$ for $i\in\{2,3\}$.
\item We define the hom-sets using $\mathcal{G}$. For example, alice's system $\kappa(a_3)\cong \mathbb{C}^2$ at $a_3$ has bob's system $\kappa(a_5)\cong \mathbb{C}^2$ at $a_5$ in its future light-cone, so the hom-set $\mathbf{T}(\kappa(a_3),\kappa(a_5))$ will be connected. Note that when $a_i$ and $a_j$ are not causally related
\[
\mathbf{T}(\kappa(a_i), \kappa(a_j))\subsetneq \mathbf{Mix}_\top (\kappa(a_i),\kappa(a_j)). 
\]
This will define a normalized subcategory of $\mathbf{Mix}_\top $. 
\item We define the partial monoidal product for $\mathbf{T}$ to exist only for disconnected systems.
\end{enumerate}

\end{example}

Hence we obtain a causal category $\mathbf{T}$ which formalizes both the resources and the causal relationships of this protocol.

\subsection{A caucat from a causet and an SMC}\label{sec:Construction}

 In this Subsection 
we shall construct a caucat from a given directed graph $\mathcal{G}=(G,\leq)$ (typically a poset) and symmetric strict monoidal category ${\bf C}$. Define $R(\mathcal{G})\subseteq 2^R$ to consist of those subsets $a\subseteq G$ satisfying  
\[
x, y\in a \Rightarrow (x\not\leq y\wedge y\not\leq x).
\] 
When tensoring we need to keep track of the space-time point an object in ${\bf C}$ is assigned to; this is achieved as follows.

\begin{defn}[Pairing construction]\label{def:construction}\em
The  \em pairing of a monoidal category ${\bf C}$ and a directed graph $\mathcal{G}$\em, is a category ${\bf CC}({\bf C}, \mathcal{G})$, defined inductively as follows:
\bit
\item Objects are either sets of pairs $\{(A_i, x_i)\}_{i\in{\cal I}}$ with $A_i\in |{\bf C}|\setminus \{I\}$ for all $i\in {\cal I}$ and $\{x_i\}_{i\in {\cal I}}\in R(G)$,  or $(I, \emptyset)$.  For $\{(A_i, x_i)\}_{i\in {\cal I}}$ and $\{(B_j, y_j)\}_{j\in {\cal J}}$ the tensor is the union and exists provided that:
\[
\{x_i\}_{i\in {\cal I}}\cap\{y_j\}_{j\in {\cal J}}=\emptyset \qquad \{x_i\}_{i\in {\cal I}}\cup\{y_j\}_{j\in {\cal J}}\in R(G)\,,
 \]
and we set $\{(A_i, x_i)\}_{i\in {\cal I}}\otimes (I, \emptyset):=\{(A_i, x_i)\}_{i\in {\cal I}}$.

\item For states we set:
\[
{\bf CC}({\bf C}, \mathcal{G})\bigl((I, \emptyset),\{(A_i, x_i)\}_{i\in {\cal I}}\bigr):=
{\bf C}(I,\otimes_{i\in {\cal I}}A_i)
\]
where due to the fact that ${\bf C}$ is symmetric the order of tensoring is not essential, just a matter of bookkeeping.
\item For general morphisms we set: 
\[
{\bf CC}({\bf C}, \mathcal{G})\bigl(\{(A_i, x_i)\}_{i\in {\cal I}}, \{(B_j, y_j)\}_{j\in {\cal J}}\bigr):=\hspace{5.4cm}
\]
\[
\left\{     
\sigma'\circ \left(  f\otimes (p\circ \top_{\otimes_{i\in {\cal I}\setminus {\cal I}'}A_i})  \right)\circ\sigma
\left| 
\!\begin{array}{c}
{\cal I}'\subseteq {\cal I}, {\cal J}'\subseteq {\cal J}\vspace{2mm}\\ 
\{x_i\}_{i\in {\cal I}'}\sqsubseteq \{y_j\}_{j\in {\cal J}'}\vspace{2mm}\\
p\in {\bf C}(I,\otimes_{j\in {\cal J}\setminus {\cal J}'}B_j)\vspace{2mm}\\
f\in  {\bf C}(\otimes_{i\in {\cal I}}A_i,\otimes_{j\in {\cal J}}B_j)\vspace{2mm}
\end{array}\!\!
\right\}\right.
\]
where $\sigma$ and $\sigma'$ are the unique symmetry isomorphisms that re-order the objects of ${\bf C}$ to match the ordering of ${\cal I}$ and ${\cal J}$, and $X \sqsubseteq Y$ means:
\[
\forall x\in X, \forall y\in Y: x\leq y\,.
\]
Finally, we close under tensoring, that is, for all $\{(A_i, x_i)\}_{i\in {\cal I}}$ and $\{(A_i', x_i')\}_{i\in {\cal I}'}$ for which the tensor exists, and all $\{(B_j, y_j)\}_{j\in {\cal J}}$ and all $\{(B_j', y_j')\}_{j\in {\cal J}'}$ for which the tensor exists,
if 
\[
f\in {\bf CC}(\{(A_i, x_i)\}_{i\in {\cal I}},\{(B_j, y_j)\}_{j\in {\cal J}})
\]
and 
\[
f'\in {\bf CC}(\{(A_i', x_i')\}_{i\in {\cal I}'},\{(B_j', y_j')\}_{j\in {\cal J}'})
\] 
then 
\[f \otimes f' \in
{\bf CC}(\{(A_i, x_i)\}_{i\in {\cal I}}\cup\{(A_i', x_i')\}_{i\in {\cal I}'},\{(B_j, y_j)\}_{j\in {\cal J}}\cup\{(B_j', y_j')\}_{j\in {\cal J}'}).
\]
\eit
\end{defn}

\begin{remark}
When defining hom-sets in Definition \ref{def:construction}, we ensured that we kept track of the space-time points with which objects are associated. To see why this is necessary, consider the following naive approach to constructing a caucat.  We set $a\leq b$ for $a, b\in R(\mathcal{G})$ when there exist $x\in a$ and $y\in b$ such that $x\leq y$.
We now attempt to define a caucat ${\bf CC}({\bf C}, \mathcal{G})$ as follows:
\bit
\item Objects are either  pairs $(A, a)$ with $A\in |{\bf C}|\setminus \{I\}$ and $a\in R(P)\setminus\emptyset$,  or $(I, \emptyset)$.
\item  Morphisms are: 
\[
{\bf CC}({\bf C}, \mathcal{G})((A,a), (B,b)):=
\left\{\begin{array}{ll}
{\bf C}(A,B)&\quad a\leq b\vspace{2mm}\\
{{\bf C}(I,B)}\circ\top_A &\quad a\not\leq b
\end{array}\right.
\]
\item The tensor of $(A,a)$ and $(B,b)$ exists iff both $a\not\leq b$ and $b\not\leq a$, or if $a$ or $b$ (or both) is the empty set, which implies that $a$ and $b$ are disjoint and that  $a\cup b\in R(\mathcal{G})$, so we can set:
\[
(A,a)\otimes (B,b):=(A\otimes B,a\cup b)\,.
\]
\eit

Consider the following scenario, as contrasted in Definition \ref{def:construction} and in the naive approach. Let $a\le b$ and $b\le c$ but $a\not\leq c$. For example, consider the spatial slices $a=\{x_1,x_2\}$, $b=\{y_1,y_2\}$ and $c=\{z_1,z_2\}$, with the relations $x_2\le y_1$ and $y_2\le z_1$, but with all other elements of the spatial slices $a,b,c$ spacelike separated. 

First we use the construction of Definition \ref{def:construction}. The objects we consider in the caucat are $\{(A_1,x_1),(A_2,x_2)\}$, $\{(B_1,y_1),(B_2,y_2)\}$ and $\{(C_1,z_1),(C_2,z_2)\}$. Suppose that in ${\bf C}$ we have connected morphisms $f:A_2\rightarrow B_1$ and $g:B_2\rightarrow C_1$. Ignoring symmetry isomorphisms (which cancel on composition), in the caucat ${\bf CC}({\bf C}, \mathcal{G})$ we have the morphisms
\[
f'=f\otimes (p_{y_2}\circ \top_{A_1}) \in {\bf CC}(\{(A_i, x_i)\}_{i\in {\cal I}},\{(B_j, y_j)\}_{j\in {\cal J}})
\]
and
\[
g'=(p_{z_2}\circ \top_{B_1})\otimes g \in {\bf CC}(\{(B_j, y_j)\}_{j\in {\cal J}},\{(C_k, z_k)\}_{k\in {\cal K}})
\]
and the composite is
\[
g'\circ f' = (p_{z_2}\circ \top_{B_1}\circ f)\otimes (g \circ p_{y_2}\circ \top_{A_1})=(p_{z_2}\otimes(g\circ p_{y_2}))\circ\top_{A_1\otimes B_1}
\]
which is disconnected, as required since $a\not\leq c$.

But consider these morphisms in the naive approach: let $A=A_1\otimes A_2$, $B=B_1\otimes B_2$ and $C=C_1\otimes C_2$; then $f'\in {\bf CC}({\bf C}, \mathcal{G})((A,a), (B,b))$ and $g'\in {\bf CC}({\bf C}, \mathcal{G})((B,b), (C,c))$.  Then $f'$ and $g'$ can be connected morphisms, and as a morphism in  ${\bf C}$, the composite $g'\circ f'$ may also be connected. Hence because the naive approach ignores how the domain and codomain of morphisms are assigned in space-time, the morphism $g'\circ f'$ may also be connected in the caucat.
\end{remark}
\begin{remark}
The causal structure (from Definition \ref{def:causal}) of the caucat ${\bf CC}({\bf C}, \mathcal{G})$ in Definition \ref{def:construction} is simply the set $R(\mathcal{G})$. 
\end{remark}

We now investigate the conditions on $\mathbf{C}$ and $P$ that make ${\bf CC}({\bf C}, \mathcal{G})$ in Definition \ref{def:construction} a caucat.

\begin{remark}\label{rmk:triv_caucat}
In  Definition \ref{def:construction}, we stipulated that the objects of  ${\bf CC}({\bf C}, \mathcal{G})$, which are sets of pairs, cannot consist of a pair $(I,a)$ unless $a=\emptyset$. This excludes the following degenerate case. Let $\mathbf{C}_0$ be the trivial monoidal category, i.e.~$|\mathbf{C}|=\{I\}$ and the only arrow is the identity morphism $1_I$. Since $\mathbf{C}_0$ has no connected morphisms, the only way to form  $\mathbf{CC}(\mathbf{C}_0, \mathcal{G})$ is for the graph $\mathcal{G}$ to be trivial, i.e.~the graph relation satisfies $x\leq y$ iff $x=y$. This follows from the fact that Definition \ref{def:construction} requires the hom-sets of causally related objects, i.e. where slices $a,b\in R(\mathcal{G})$ satisfy $a\sqsubseteq b$, to be connected.
\end{remark}

\begin{example}
Let $\mathbf{C}_1$ be the monoidal category with only one non-trivial object $A\neq I$, and only one disconnected morphism $f:A\rightarrow A$. Let the graph $\mathcal{G}$ have elements $G=\{x,y\}$, with the order $x\leq y$. Then ${\bf CC}({\bf C}, \mathcal{G})$ has one connected morphism $f:(A,x)\rightarrow (A,y)$. This also shows that the category ${\bf CC}({\bf C}, \mathcal{G})$ does not exist for an arbitrary monoidal category (even if non-trivial) and graph, since if $f$ is disconnected in this example then $\mathcal{G}$ must be the trivial graph as in Remark \ref{rmk:triv_caucat}.
\end{example}

This leads to the following proposition.

\begin{prop}
For any digraph $\mathcal{G}$, the category $\mathbf{CC}(\mathbf{C}_1, \mathcal{G})$ is a caucat. 
\end{prop}


\begin{example}
Consider again the category $CPM_\top({\bf C})$ that was introduced in Example \ref{ex:cpm}. In this category $I$ is terminal:  if $\pi:A\to I$, then $\pi=1_I\circ \pi=\top_I\circ \pi=\top_A$.
\end{example}

\begin{lem}\label{lem:pairing}
The pairing ${\bf CC}({\bf C}, \mathcal{G})$ of $\mathbf{C}$ and $\mathcal{G}$ is a normalized category iff $\mathbf{C}$ is a normalized category.
\end{lem}
\begin{proof}
We need to show that the unit object of ${\bf CC}({\bf C}, \mathcal{G})$ is terminal iff the unit object of $\mathbf{C}$ is terminal. By the definition of the pairing construction the unit object in  ${\bf CC}({\bf C}, \mathcal{G})$ is $(I, \emptyset)$. For any object $(\{(A_i, x_i)\}_{i\in {\cal I}}$, the morphisms to the unit object are defined as 
\[
{\bf CC}({\bf C}, \mathcal{G})\bigl(\{(A_i, x_i)\}_{i\in {\cal I}}, (I,\emptyset)\bigr):=\mathbf{C}(\otimes_{i\in {\cal I}}A_i,I).
\]
Hence the unit $I$ in $\mathbf{C}$ is terminal iff the unit $(I,\emptyset)$ in ${\bf CC}({\bf C}, \mathcal{G})$ is.
\end{proof}

\begin{prop} The pairing  ${\bf CC}({\bf C}, \mathcal{G})$ of $\mathbf{C}$ and $\mathcal{G}$ is a caucat iff  $\mathbf{C}$ is normalized and ${\bf C}(I,A)$ is non-empty for all $A\in |\mathbf{C}|$.
\end{prop}
\begin{proof}
(If) A normalized category is a caucat if it has a partial tensor whose existence is defined by disconnected hom-sets, and if each object has at least one state. By Lemma \ref{lem:pairing}, if $\mathbf{C}$ is normalized then the pairing ${\bf CC}({\bf C}, \mathcal{G})$  is normalized. Also, Definition \ref{def:construction} ensures both that the monoidal product of two objects in ${\bf CC}({\bf C}, \mathcal{G})$ exists only when they are disconnected, and that each object in ${\bf CC}({\bf C}, \mathcal{G})$ has at least one state if each object in $\mathbf{C}$ has at least one state. 

(Only if) If  ${\bf CC}({\bf C}, \mathcal{G})$  is a caucat then by Lemma \ref{lem:pairing}, the category $\mathbf{C}$ is normalized. Also, by  Definition \ref{def:construction}, each object in ${\bf CC}({\bf C}, \mathcal{G})$ has at least one state only if each object in $\mathbf{C}$ has at least one state. 

\end{proof}

\subsection{Recovering CQM}\label{sec:recover}

The results of Subsection \ref{sec:Incompatibilities} showed that key structures of CQM, such as compactness and the dagger functor, cannot be retained in causal categories. These structures are used in CQM to formalize notions such as measurement and dynamics. However, we can still reinstate the power of CQM by providing a precise description of the connection between a caucat and the dagger compact category that is typically used to construct it. Indeed, the constructions described above provide for the existence of a functor from the caucat to this dagger compact category of processes. 

In the first case, a carved category $\mathbf{CC}$ is a subcategory of a given $\dagger$-SMC $\mathbf{C}$, and so has an embedding functor 
\[
\iota:\mathbf{CC}\rightarrow \mathbf{C}.
\]
In the case of the second construction, a paired category  ${\bf CC}({\bf C}, \mathcal{G})$ has a projection functor 
\[
\pi: {\bf CC}({\bf C}, \mathcal{G})\rightarrow {\bf C}.
\]
We now show in what way 
these functors can be exploited to recover CQM.

\begin{example}
We can identify unitary operations in a caucat as follows. Consider a carved category $\mathbf{CC}$ of a $\dagger$-SMC $\mathbf{C}$. We have the inclusion functor $\iota:\mathbf{CC}\rightarrow \mathbf{C}$, from which we can define \em causal unitary \em operations in $\mathbf{CC}$ using conditions on the image of $\iota$. Since $\iota(U)=U$, $\iota(A)=A$ and $\iota(B)=B$, causal unitary morphisms  $U:A\rightarrow B$ in $\mathbf{CC}$ are those which are unitary in $\mathbf{C}$, i.e.~those which satisfy
\[
U^\dagger\circ U=1_ A\;\; \& \;\; U\circ U^\dagger=1_ B. 
\]
in $\mathbf{C}$. A general (i.e.~not necessarily dagger-) isomorphism is similarly associated with a morphism in $\mathbf{CC}$, by witnessing it in $\mathbf{C}$.
\end{example}

Now, we showed in Proposition \ref{thm:isos} that the fact that a system remains of the same type (e.g.~a qubit) during an evolution process cannot be defined in the caucat itself. However the above example shows that this can instead be defined using the functor $\iota$ and data in $\mathbf{C}$. 

\begin{example}
Measurements can be identified in a caucat in a similar way to unitaries, although some care is needed due to the same issue of isomorphisms. A morphism $v:A\rightarrow B\otimes C$ in $\mathbf{CC}$ is a candidate for representing a measurement, since its codomain can represent a state space $B$ and a space of classical outcomes $C$. However, measurements in $\mathbf{C}$ are usually defined with identical types for the input and output state spaces, i.e.~it has the form $v:A\rightarrow A\otimes C$. We account for this by requiring that instead the input and output types are isomorphic, where, as discussed in the previous example, we define this isomorphism in $\mathbf{C}$. So one of the conditions for a morphism $v:A\rightarrow B\otimes C$ to be considered a measurement in  $\mathbf{CC}$, is that $B\cong A$ holds in $\mathbf{C}$.  

Consequently, for  a carved category  $\iota:\mathbf{CC}\rightarrow \mathbf{C}$, a  \em causal measurement \em is a morphism $v:A\rightarrow B\otimes C$ in $\mathbf{CC}$ such that
\bit
\item there is an isomorphism $f:B\cong A$ in $\mathbf{C}$
\item denoting the morphism $m:=(f\otimes 1_C)\circ v:A\rightarrow A\otimes C$ as 
\[
\epsfig{figure=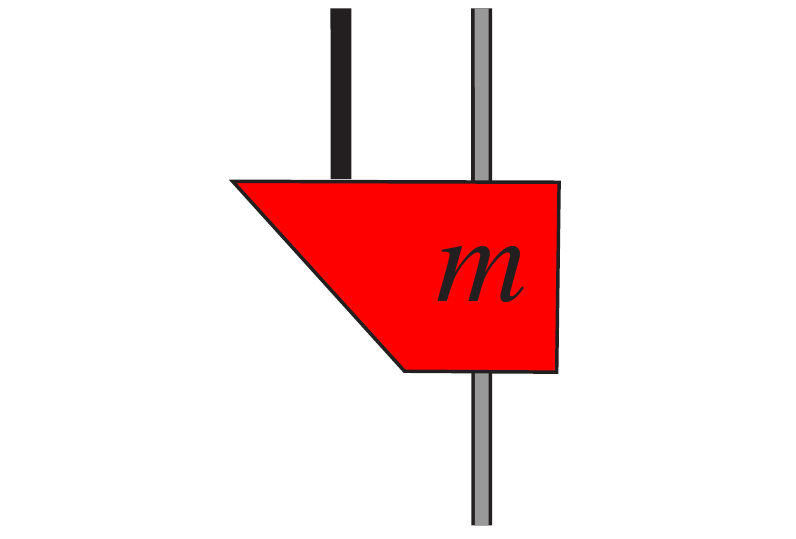,width=75pt}
\]
the morphism $m$ satisfies the (non-destructive) measurement axioms defined for morphisms in a $\dagger$-SMC \cite{CPav}:
\[
\epsfig{figure=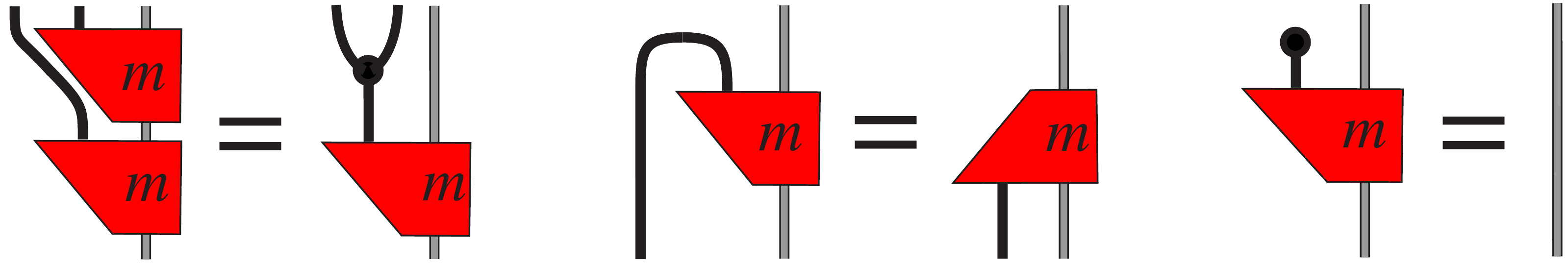,width=300pt}
\]
As discussed in \cite{CPav}, these axioms define an abstract notion of measurement for a dagger compact category, and in $\mathbf{FdHilb}$ they define a family $\{P_i\}_i$ of mutually orthogonal projectors. In that case the axioms respectively correspond to the von Neumann projection postulate (viz.~idempotence of $P_i$), self-adjointness of $P_i$, and completeness of  $\{P_i\}_i$, i.e.~$\displaystyle\sum_i P_i=1_\mathcal{H}$.
\eit
\end{example}


\section{Conclusion and outlook}


We have introduced causal categories as a mathematical vehicle which explicitly encodes causal structure within a category of processes that may take place within that causal structure.  The axioms of this structure were motivated by:
\bit
\item[(1)] compatibility of this causal structure and signaling processes;
\item[(2)] well-definedness of the concept of global state for all objects.
\eit
Causal categories are themselves typically neither compact nor dagger, but the constructions that we defined always provide a causal category with an ambient dagger compact category, hence retaining the full power of categorical quantum mechanics.

The unique morphisms that witness terminality of the tensor unit---terminality follows from (1)---also appeared in an axiomatization of mixedness in \cite{CPer}, which imposes additional structure on an environment structure.  It would be worthwhile to understand the interaction of this additional structure with the causal category structure better.  As this structure involves a purification condition, one would also like to understand the connection with the work by Chiribella-D'Ariano-Perinotti \cite{CDP1, CDP2} better.

%
%
%
\appendix
\section{Appendix}\label{sec:appendix}

In this Appendix we expand on the assumptions made at the end of Section \ref{sec:monoidal}, that symbolic formulae are stated up to equivalence in the diagrammatic calculus. Our aim is to show how connectedness in the graphical language can be formalised symbolically. 

The motivation for this is as follows.
Firstly we recall from Section  \ref{sec:monoidal} that graphical connectedness distinguishes the two modes of composition symbolically, except when the tensor unit is involved.  However, the two modes of composition have diffferent, complementary, physical interpretations: parallel composition means independence. But, as we showed in Section \ref{sec:CausalityInSMC}, graphical disconnectedness also means independence. Since graphical disconnectedness subsumes parallel composition (as we just described), it also is the \em primary \em concept of independence. But the graphical language arose in a formal way from considering SMCs (see Theorem \ref{thm:SMCcoherence}), which is usually defined symbolically. Hence it is desirable to understand what kind of symbolic language corresponds to this graphical language of connectedness: this will provide a `syntax for information flow'.

In other words, we want to `fix' the language so that formulae faithfully reflect the connectedness of a corresponding diagram. 

\paragraph{Standing assumptions about formulae.}
Since it is the  non-atomic morphisms that would allow one to `hide' disconnectedness, our two key \em standing assumptions\em in Section \ref{sec:monoidal} were:
\begin{enumerate}
\item all morphism formulae contain only atomic morphisms;
\item a morphism is expressed as a formula using $\otimes$ instead of $\circ$ whenever possible.
\end{enumerate}
This corresponds to:
\begin{enumerate}
\item ensuring that the `atoms' of our language are connected morphisms, so that sequential composition corresponds to genuine information-flow
\item revealing genuine disconnectedness of a morphism in its corresponding formula (i.e. a formula reveals non-atomicity of the morphism). 
\end{enumerate}
In short, we  force the syntax of the morphism language to faithfully represent the topology of the graphical language.

\begin{remark}[Morphism language formulae containing only atomic morphisms]
Note that the assumption that all  formulae contain only atomic morphisms can be made without any loss of generality, because we only consider protocols that involve a finite number of systems.  An arbitrary formula can be turned into one made up of atomic morphisms simply by substituting all non-atomic morphisms by a non-trivial parallel composition, and repeat this procedure until all  morphisms are atomic. Without finiteness this procedure may not terminate. This can be seen as a formally  distintinguishing between the following two views on the relation between a system and its subsystems. This first view is `compositional': we describe a system by considering jointly its subsystems, and the latter are conceptually prior. Conversely, the second view is `decompositional': Nature appears to us as a single entity, and we can tractably do physics by decomposing it into susbsystems; now the global system is conceptually prior. 
.
\end{remark}

\begin{defn}[`Disconnectedness' for symbolic language]\label{def:factors}\em
We say that a symbolic language expression $\mathcal{F}:\mathcal{A}\rightarrow\mathcal{B}$ \em factors through $\otimes$ \em if, by pre- and post-composition of generalized symmetry morphisms and
by using the axioms of SMCs, $\mathcal{F}$ cannot be rewritten as a symbolic language expression $\mathcal{F}_1\otimes\mathcal{F}_2$.
\end{defn}
For example, given $\top:B\rightarrow I$ and $\psi:I\rightarrow B$, the expression $(g\circ f)\otimes(\psi\otimes\top)$ does not factor through $\otimes$, whereas the expression $(1_I\otimes\psi)\circ(\top\otimes 1_I)$ does, because the axioms of SMCs ensure that 
\[
(1_I\otimes\psi)\circ(\top\otimes 1_I) = (1_I\circ\top)\otimes(\psi\circ1_I).
\]
Note that Definition \ref{def:factors} captures a notion that is distinct from non-atomicity, since the latter refers to whether a morphism is equal to the tensor of two other morphisms (and is therefore about the contingent equations of the category), whereas the former refers to whether a morphism formula can be manipulated syntactically  to a particular form, using the axioms of symmetric monoidal categories. The requirement of pre- and post-composition in Definition \ref{def:factors} is included because of morphisms such as the symmetry isomorphism $\sigma_{A,B}:A\otimes B\rightarrow B\otimes A$. This will be symbolically disconnected because post-composition with $\sigma_{B,A}$ yields $1_A \otimes 1_B$.

\begin{thm}[Relationship between symbolic and graphical connectedness]\label{thm:syntax}
Given the standing assumptions above, the topological disconnectedness of a morphism in the graphical language is equivalent to its corresponding symbolic language expression factoring through $\otimes$.
\end{thm}
\begin{proof}
This follows by Theorem \ref{thm:SMCcoherence}  \cite{JS}, which states that the graphical calculus corresponds fully and faithfully to the syntactic manipulations of the corresponding symbolic language expressions. 
\end{proof}
The significance of Theorem \ref{thm:syntax} is that it provides a formal expression of
\begin{quotation}
reading the topology from the syntax.
\end{quotation}




\end{document}